%% file: main.tex
\documentclass[11pt]{article}

\usepackage[utf8]{inputenc}
\usepackage{fullpage}
\usepackage[ruled,vlined,linesnumbered]{algorithm2e}
\usepackage{graphicx}
\usepackage{amssymb,amsmath,amsfonts}
\usepackage{amsthm}
\usepackage{multicol}
\usepackage[T1]{fontenc}
\usepackage{color}
\usepackage{authblk}

\newcommand{\dk}[1]{#1}
\newcommand{\cb}[1]{#1}
\newcommand{\cbr}[1]{#1}
\newcommand{\acpd}{{\sf Asymmetric-C-Partial-D}}
\newcommand{\scfd}{{\sf Symmetric-C-Full-D}}
\newcommand{\combined}{{\sf Combined-C-D}}
\newcommand{\e}{{\epsilon}}
\newcommand{\remove}[1]{}

\newtheorem{theorem}{Theorem}[section]

\newtheorem{lemma}{Lemma}

\begin{document}

\title{Byzantine-Resilient Population Protocols
}

\author{Costas Busch}
\author{Dariusz R. Kowalski}
\affil{
\textit{School of Computer \& Cyber Sciences,
Augusta University}\\
Augusta, Georgia, USA \\
\{kbusch,dkowalski\}@augusta.edu
}





\date{}
\maketitle

\begin{abstract}
Population protocols model 
information spreading and computation in 
network systems
where pairwise node exchanges are determined by an external random scheduler.
Moreover, nodes are often assumed to have small memory and computational resources.
Most of the population protocols in the literature assume 
that the 
participating
$n$ nodes are honest. 
Such an assumption may not be, however, accurate for large-scale systems of small devices. Hence,
in this work,
we study 
population protocols
in a setting
where up to $f$ 
nodes can be Byzantine.
We examine the majority (binary) consensus problem
against
different levels of adversary strengths,
ranging from the Full Byzantine adversary that has complete knowledge of all the node states
to the Weak Content-Oblivious Byzantine adversary that has only knowledge about which exchanges take place.
We also take into account 
Dynamic vs Static node corruption by the adversary.
We give lower bounds that require any algorithm 
solving
the majority consensus 
to have initial difference $d = \Omega(f + 1)$ for the tally between the two proposed values,
which holds
for both the Full Static and Weak Dynamic adversaries.
We then present an algorithm that solves the majority consensus problem 
and
tolerates $f \leq n / c$ Byzantine nodes, for some constant $c>0$, with
$d = \Omega(f + \sqrt{n \log n})$ and $O(\log^3 n)$ parallel time steps,
using $O(\log^3 n)$ states per node.
We also give an alternative algorithm, with the same asymptotic performance, for 
$d = \Omega(\min\{f \log^2 n + 1,n\})$.
Finally,
we combine both algorithms into one using a new robust distributed common coin.
The only other known previous work on Byzantine-resilient population protocols 
tolerates up to $f = o(\sqrt n)$ faulty nodes and 
was analyzed
against a Static adversary;
hence, our protocols significantly improve 
fault-tolerance by an $\omega(\sqrt n)$~factor and all of them work correctly against a stronger Dynamic adversary.
\end{abstract}




\input{introduction}
\input{model}
\input{lowerbound}
\input{skeleton}
\input{algolog3}
\input{algoprecise}
\input{combined}
\input{related}

\input{conclusion}


\bibliographystyle{abbrv}

\input{literature.tex}





 
\end{document}

%% file: introduction.tex
\section{Introduction}

Consider a system of $n$ nodes such that each node has a finite local state. 
The nodes do not have identifiers, hence, any two nodes in the same state are indistinguishable.
The global state of the system (configuration) is the collection of the individual local states of all the nodes. Two nodes communicate with each other through an {\em exchange} that carries information from one node to the other in both directions between them. 
The exchange may alter the local states of the involved nodes, and hence, 
an exchange may alter the global state as well.
An exchange scheduler specifies a sequence of exchanges in the system.
Population protocols consider the {\em randomized scheduler} 
that 
selects a pair of nodes to perform 
an exchange between them uniformly at random~\cite{AngluinADFP06,berenbrink2018population}.
\cbr{Random exchanges model interactions in natural (i.e. biological or chemical) systems where nodes meet randomly and exchange information.}

In the {\em binary majority (consensus) problem}, each node 
in the  initial configuration has an initial preference value which is either $A$ or $B$.
The task is to perform exchanges to converge to a final configuration where every node 
chooses the same value ($A$ or $B$) that represents the 
majority of the initial preferences.
A performance metric for a majority consensus algorithm is the 
{\em parallel time} to converge to the decision, 
which is the total number of exchanges until convergence divided by $n$.
Alistarh {\em et al.}
\cite{DBLP:conf/podc/AlistarhGV15} give a $\Omega(\log n)$ lower bound on 
the parallel time for the majority problem.
In the literature, several works have contributed to the problem of 
solving the majority consensus problem in poly-logarithmic parallel time steps, and also with small number of states, namely, 
polylogarithmic number of
states per node or less \cite{AlistarhAEGR17, DBLP:conf/podc/AlistarhGV15, berenbrink2018population, DBLP:journals/dc/MertziosNRS17}.
Note that a polylogarithmic number of states can be represented by $O(\log\log n)$ bits, therefore our algorithms use $O(\log\log n)$ bits for both local memory~and~exchange. 

In the traditional setting, population protocols assume that all nodes are honest 
and follow the steps of the protocol in order to accomplish the joint consensus task.
On the other hand, they aim at modelling a behavior of a large-scale system of ``small'' units, which in practise could malfunction and be difficult to ``fix manually''. Therefore,
we consider here an alternative version of the majority problem involving a Byzantine adversary 
that may corrupt up to $f$ nodes, which will refer to as the {\em faulty nodes}.
The Byzantine adversary fully controls the faulty nodes and their states, and may attempt to destabilize the majority algorithm,
by delaying or disabling 
its convergence.
Moreover,
the Byzantine adversary may cause the 
minority initial preference to be decided by the honest~nodes,
diverting the~correct~decision.

More precisely, a Byzantine node can only affect the state of an honest node during an exchange, by forcing certain transitions to the state machine of the honest node with the hope to derail the majority consensus process. However, the honest nodes are still able to operate after the exchanges with the faulty nodes, and with appropriate algorithm they could even compute a majority value correctly. We consider the following types of Byzantine Adversaries (B.A.):
\begin{itemize}
    \item {\em Full (dynamic) B. A.:}
    knows the algorithm, observes nodes and their states,  corrupts up to $f$ nodes in online fashion.
    \item {\em Weak (dynamic) B. A.:}
    knows the algorithm,  observes choices of the scheduler and states resulting from exchanges (but not the future),
    corrupts up to $f$ nodes in online fashion.
    \item {\em Weak (dynamic) Content-Oblivious B. A.:}
    knows the algorithm,  observes only the choices of the scheduler (but not the future), corrupts up to $f$ nodes in online fashion.
    \item {\em Full/Weak Static B. A.:} Similar to the Full/Weak dynamic B. A., but corrupted nodes are selected before the algorithm starts.
\end{itemize}
Note that the Static 
adversary is not more powerful than the Dynamic, 
and the Weak is not more powerful than the Full, 
if the other characteristics 
are the same.
Here by ``not more powerful'' we mean that it cannot create distributions of executions different~than~the~other~one can.

Let $d$ denote the absolute difference between the number 
of nodes that have initial preference $A$ and $B$.
The only other known population protocol 
that is Byzantine-resilient is given by Angluin {\em et al.}
\cite{DBLP:journals/dc/AngluinAE08} which present a 3-state population protocol
that requires $d = \omega({\sqrt n} \log n)$
and converges to the majority value in $O(\log n)$ parallel time.
The protocol tolerates up to $f = o(\sqrt n)$ faulty nodes.
The approach of the protocol is that exchanges involving opposite value nodes 
result in nodes canceling their preferences becoming blank nodes,
and nodes that become blank eventually adopt the preference of other nodes in future exchanges.
Using a potential function analysis they show that the algorithm converges
to a state where all nodes correctly adopt the majority value assuming the initial difference $d$.

\cbr{Inspired by classic distributed consensus algorithms,
where there are known protocols that 
can tolerate up to $f < n/3$ faulty nodes~\cite{DBLP:journals/toplas/LamportSP82,DBLP:journals/jacm/PeaseSL80},
we show that similar asymptotic tolerance can also be achieved in 
population protocols.
Compared to the previous work \cite{DBLP:journals/dc/AngluinAE08},
our protocols tolerate larger number of faulty nodes} 
which is up to a constant fraction of $n$, that is, $f = O(n)$.
Next, the tally difference $d$ handled by our protocols is proportional to the number of failures $f$, for $f=\Omega(\sqrt{n\log n})$.
Moreover, the protocol in~\cite{DBLP:journals/dc/AngluinAE08} 
was analyzed only
against a static adversary, where the faulty nodes 
are chosen before the execution, while ours is for the stronger and more general dynamic adversary,
where the faulty nodes are determined during the algorithm execution.
The trade-off for the higher tolerance is that we use $O(\log^3 n)$ states per node,
and a 
$O(\log^2 n)$
factor slowdown in convergence time.
Table \ref{tab:comparison} has a comparison of the results in \cite{DBLP:journals/dc/AngluinAE08}
with~ours.

Berenbrink et al.~\cite{berenbrink2018population}
give a population protocol with $O(\log n)$ states that converges 
to the majority value in $O(\log^{5/3} n)$ parallel time and only requires $d = 1$.
This result is further improved to $O(\log^{3/2} n)$ parallel time~\cite{10.1145/3382734.3405747}.
However, these works are not Byzantine-resilient like ours.
Nevertheless, we use similar techniques as in~\cite{berenbrink2018population} for the protocol design and analysis 
which we adapt appropriately to develop the fault resiliency.
Note that our protocols also operate when $f=0$ and compared to 
the protocol in \cite{berenbrink2018population} 
the achieved performance is within a poly-log factor 
with respect to time using the same number of states and $d = 1$.
There have also been other works on regular population protocols (non-Byzantine-resilient),
which trade-off between $d$, parallel time, and number of states.
A broader discussion on related work appears in Section~\ref{sec:related}.

\vspace*{1ex}
\noindent
{\bf Motivation:}
Byzantine-resilient distributed computing, and the consensus problem in particular, has been in the center of distributed computing~\cite{Attiya-Welch-book2004} and systems~\cite{tanenbaum2017distributed} (Chapters 8 and 9) since the seminal works by Lamport {\em et al.}~\cite{DBLP:journals/toplas/LamportSP82,DBLP:journals/jacm/PeaseSL80}. It also gathers attention from the security community, c.f.,~\cite{DBLP:conf/stoc/CanettiR93}. Our study, similarly to vast majority of research on Byzantine-resiliency, uses mathematical tools to analyze properties and provide bounds on the performance of the proposed algorithms, as it is computationally too expensive to accurately simulate Byzantine behavior in asynchronous systems. However, our protocols could be easily and efficiently implemented in almost any distributed system, as (1) they use very limited resources and very simple communication mechanism, (2) they work in a presence of asynchrony and Byzantine failures, and (3) they are provably very fast even in worst-case executions.

\begin{table*}[t]
    \centering
    \begin{tabular}{|c|c|c|c|c|}
       \hline
       {\bf Population Protocol}  & {\bf $f$} & {\bf Tally Difference $d$} & {\bf Time} & {\bf States}\\
       \hline
       \cite{berenbrink2018population} & 0 & $\Omega(1)$ & $O(\log^{5/3} n)$ & $O(\log n)$\\
       \hline
       \ \ \! \cite{DBLP:journals/dc/AngluinAE08} $^*$ & $o(\sqrt n)$ & $\omega({\sqrt n} \log n)$ & $O(\log n)$ & $O(1)$\\
       \hline
       \acpd & $O(n) $ &$\Omega(f + \sqrt{n \log n})$ 
       & $O(\log^3 n)$& $O(\log^3 n)$\\
       \hline
       \scfd & $O(n) $ & $\Omega(\min\{f \log^2 n + 1,n\})$ & $O(\log^3 n)$ & $O(\log^3 n)$\\
       \hline
       \combined & $O(n) $ & $\Omega(\min\{f + \sqrt{n \log n}, f \log^2 n + 1, n\})$ & $O(\log^3 n)$ & $O(\log^3 n)$\\
       \hline
    \end{tabular}
    
    \caption{Results and comparison with previous works. $^*$ The protocol in \cite{DBLP:journals/dc/AngluinAE08} is designed for a static adversary, while our protocols are 
    designed and analyzed
    for the stronger and more general dynamic adversary. 
    Note that if $f>0$, the lower bound on the tally difference $d$ is $\Omega(f+1)$, by Lemmas~\ref{lem:lower-full},~\ref{lemma:lower-weak-adaptive},~\ref{lemma:lower-weak-oblivious}.}
    \label{tab:comparison}
\end{table*}

\subsection{Contributions}
We present lower bounds and algorithmic upper bounds (main results) for the binary majority consensus problem with respect to the Byzantine adversary models.
We assume the fully connected graph model where the randomized adversarial scheduler picks each time 
a random pair of nodes to perform an exchange.
A comparison of our results with the previous work can be found in Table \ref{tab:comparison}.

\vspace*{1ex}
\noindent
{\bf Lower Bounds:}
Let $a$ and $b$ be the tally of the nodes having initial value $A$ and $B$, respectively. Let $d = |a - b|$.
We show that any algorithm for the Full Byzantine adversary, even Static,
requires 
$d\ge 2f$
in order for the honest nodes to make the right decision on the majority value. This holds for both deterministic and randomized 
algorithms, c.f., Lemma~\ref{lem:lower-full}.

We also show a similar lower bound for the Weak (Dynamic) Byzantine adversary, where we show that any algorithm requires $d\ge 2(f-1)$ for a vast majority of values~$f$, c.f., Lemma~\ref{lemma:lower-weak-adaptive}.
Slightly weaker criteria of $d\ge f-1$ can be enforced by the Weak Content-Oblivious adversary, see Lemma~\ref{lemma:lower-weak-oblivious}.
    
\vspace*{1ex}
\noindent
{\bf Upper Bounds:}
We present three algorithms that are resilient to the strongest considered {\em Full Dynamic Byzantine Adversary} and tolerate $f \leq n/c_f$, for some constant $c_f$. The first two algorithms ({\acpd} and \scfd) have {\em deterministic transition functions} for the local node states and make different assumptions about $d$. The third algorithm (\combined) combines the two algorithms into a single algorithm with a unified assumption~for~$d$. 
\begin{itemize}

\item Algorithm \acpd:
This algorithm works for $d = \Omega(f + \sqrt{n \log n})$.
It operates with asymptotically optimal difference $d$ when $f = \Omega(\sqrt{n \log n})$.
The asymptotic performance of the algorithm is stated in the following theorem.

\end{itemize}

\begin{theorem}
\label{theorem:algo1}
For sufficiently large $n$, 
and $f \leq n/c_f$, for some constant $c_f$,
if $d = \Omega(f + \sqrt{n \log n})$
then with Algorithm \acpd\ all honest nodes decide the majority value
in $O(\log^3 n)$ parallel time steps using $O(\log^3 n)$ states per node
with probability at least $1 - O(\log n) / n$.  
\end{theorem}

\begin{itemize}
    
\item Algorithm \scfd: 
This algorithm works for $d = \Omega(\min\{f \log^2 n + 1,n\})$. 
The algorithm has the benefit that it works also for the case when $d = o(\sqrt{n \log n})$,
with the trade-off that in the expression of $f$ the term $d$ is multiplied by a poly-log factor of $n$.
The asymptotic performance of the algorithm is stated in the following theorem.

\end{itemize}

\begin{theorem}
\label{theorem:algo2}
For sufficiently large $n$, 
and $f \leq n / c_f$, for some constant $c_f$,
if $d = \Omega(\min\{f \log^2 n + 1,n\})$
then with Algorithm \scfd\ all honest nodes decide the majority value
in $O(\log^3 n)$ parallel time steps using $O(\log^3 n)$ states per node 
with probability at least $1 - O(\log^3 n) / n$.  
\end{theorem}
    
\begin{itemize}
    
\item Algorithm \combined:
This algorithm combines the bounds of the two Algorithms \acpd\ and \scfd\
to give a unified $d$ in terms of $f$.
But since $f$ is not known, the algorithm has to adjust to one of the two algorithms that offers better performance for this $f$ -- it is done by implementing and using a random bias. 
Therefore, the transition function of the nodes is {\em randomized}.

\end{itemize}

\begin{theorem}
\label{theorem:combined}
For sufficiently large $n$, 
and $f \leq n / c_f$, for some constant $c_f$, 
if $d = \Omega(\min\{f + \sqrt{n \log n}, f \log^2 n + 1, n\})$
then with Algorithm \combined\ all honest nodes decide the majority value
in $O(\log^3 n)$ parallel time steps using $O(\log^3 n)$ states per node with probability 
at least $1 - O(\log^3 n) / n$.
\end{theorem}


\noindent    
{\bf Techniques and novelty:}
Our first two protocols are based on dividing the execution steps of the nodes into three distinct phases: cancellation, duplication, and resolution.
Although this high-level approach may look
similar to~\cite{berenbrink2018population},
our implementations of the phases are substantially different, as the previous protocols
were not inherently designed to tolerate Byzantine faults
and
thus
may return the wrong outcome in the presence of faulty nodes.

Our main idea of
the cancellation phase 
is that
the nodes 
participating
in an exchange and 
having
opposite values 
become {\em empty}, holding no value. 
On the other hand, in 
our implementation of the
duplication phase, the non-empty nodes (holding value $A$ or $B$) clone their values on the empty nodes.
This has the effect of amplifying (doubling) the 
difference $d$ between the tallies of the two values.
By repeating this a logarithmic number of times, 
the difference $d$ of the tallies becomes linear eventually.
At 
that
point, the nodes attempt to make a decision in the resolution phase
by deciding the value that is sampled the most.
The duration of each phase is a poly-logarithmic number of parallel steps, 
giving a poly-logarithmic overall time.

With the Byzantine adversary, several things can go wrong during this process.
During the cancellation phase, the Byzantine adversary 
may
attempt to cancel values of the majority,
decreasing the difference of tallies $d$ and possibly converting the minority value to a majority.
During the duplication phase, the Byzantine adversary 
could try
to clone values of the minority,
again having the same negative effect as in the cancellation phase.

To limit such negative behavior,
the first algorithm (\acpd) uses another new concept of {\em Asymmetric Cancellations} in the exchanges,
where in the exchange only one side unilaterally cancels its value while the other side does not.
A cancellation exchange executes only once in a phase per honest node.
This limits significantly the negative impact of the Byzantine nodes to the honest nodes.
Similarly, in the duplication phase, there is a single 
chance for an empty honest node to participate in an exchange that can make it 
adopt a cloned value in the phase.
This may not allow to clone all possible non-empty node values,
resulting 
in
{\em Partial Duplication} (less than doubling factor).
Because of the asymmetric cancellations and partial duplication,
this algorithm requires an initial tally difference $d = \Omega(f + \sqrt{n \log n})$
in order to work properly and give a majority decision.
However, the benefit is that for {\em large} $f = \Omega(\sqrt{n \log n})$
the parameter $d$ is asymptotically optimal.

On the other hand, the second algorithm (\scfd) 
uses somehow complementary approach of {\em Symmetric Cancellations} and 
{\em Full Duplication}.
This is because this algorithm, unlike \acpd, behaves better for a {\em smaller} number of faulty nodes $f$.
The nodes execute multiple cancellation or duplication exchanges 
per phase,
which exposes them for longer periods of time to the negative impacts of the Byzantine nodes.
For this reason, the algorithm requires an asymptotically larger initial difference $d = \Omega(f \log^2 n + 1)$, which inserts a $\log^2 n$ factor to the $f$ term,
but has the benefit 
of avoiding
the additive $\sqrt{n \log n}$~term.

The last algorithm (\combined)
uses local random coins to combine both algorithms into one, which works efficiently for {\em all} values of $f$ even without knowing $f$. 
To correctly guess which of the two algorithms produces the correct output,
it creates a $\Theta(\sqrt{n \log n})$ {\em random bias} towards one of the input values, which we call a random common coin. 
It runs the algorithms three times, once with the original initial values, 
and the next two runs with the random bias toward one of the input values each time.
Depending on the outcomes of the runs, it makes 
a correct prediction of
which algorithm produced the right output.

\subsection*{Outline of Paper}
We present our lower bounds 
for the Full and Weak Byzantine adversaries in Section~\ref{sec:lower-bounds}.
\cbr{In Section \ref{section:skeleton}, we present a skeleton algorithm for determining the majority value
with a basic analysis.}
We present Algorithm \acpd\ and its analysis in Section~\ref{sec:acpd}
while Algorithm \scfd\ appears in Section~\ref{sec:scfd}.
We give the combined algorithm description and analysis in Section~\ref{sec:combined}.
Related work appears in Section \ref{sec:related}, and
conclusions in Section~\ref{sec:conclusion}.

%% file: model.tex
\section{Computation Model}
\label{sec:model}

\subsection{Basic Population Protocol Model}
Our population protocol model is based on the models 
in
\cite{berenbrink2018population}
and \cite{DBLP:journals/dc/AngluinAE08}, but adapted to incorporate appropriately the Byzantine adversary~models. 

We assume a set of {\em nodes} $V$, where $|V| = n$,
which are connected in a {\em complete undirected simple graph}.
The nodes do not have any identifiers distinguishing them from others.
An execution of a population protocol, 
run
autonomously at nodes, proceeds in consecutive steps.
In an {\em execution step}, a pair of nodes $(u,v) \in V \times V$ is chosen uniformly at random from the set of all 
${n \choose 2}$
node pairs to participate in an {\em interaction (exchange)}.
Similarly to \cite{AlistarhAEGR17}, we assume that the chosen pair is undirected, as we do not distinguish between initiator and responder nodes in the pair. 
Due to 
this
random scheduler of exchanges, all executions are random.
\cbr{
The exact mechanism to implement an 
exchange depends on the actual system. 
If the system is a network of nodes then an exchange can 
correspond to sending messages with the nodes' state along network links. 
In a biological system the nodes are cells
where an exchange is a chemical/biological interaction between the cells.
}

Each node has a {\em state} taken from a set of states $S$.
A {\em configuration} is a collection of the states of all the nodes at 
the beginning or the end of any specific step.
Each exchange affects the states of the involved nodes
according to a joint {\em transition function} \dk{(also called an {exchange function})} $\Delta: S \times S \to S \times S$.
Namely, when an exchange pair $(u,v)$ is chosen with current states $s_u, s_v$,
the respective states after the exchange become $s'_u, s'_v$,
where $\Delta(s_u,s_v) = (s'_u, s'_v)$.
Note that, similarly to \cite{AlistarhAEGR17}, 
the transition function $\Delta$ is symmetric, that is, it also holds that $\Delta(s_v, s_u) = (s'_v, s'_u)$, since there is no notion of direction on the chosen interacting pair.
Function $\Delta$ is deterministic in case of deterministic algorithms, and random in case of randomized algorithms.
Nodes that do not take part in the interaction do not change their states at that step.


A population protocol also involves input and output functions 
related to the computation task.
There is a finite set of input symbols $X \subseteq S$, 
a finite set of output symbols $Y\subseteq S$,
and an output function $\Gamma: S \to Y$.
The system is initialized with an input function $\Lambda$ 
which maps every node to an input symbol,
that is, $\Lambda: V \to X$.
In population protocols, for the {\em majority problem} the input and output symbols are $X = Y = \{A,B\}$, i.e.,
the input function $\Lambda$ initializes each node to one of the two symbols, $A$ or $B$.
Function $\Lambda$ is always deterministic, even in case of randomized solutions, since it is a problem-related mapping.

We say that the protocol reaches a {\em stable configuration} at step $t$ if for each node it holds that 
for each step after $t$ the output of the node remains the same as its output immediately after the execution of step $t$.
Note that after reaching a stable configuration in step $t$ a node may reach different states, but all those states result to the same output, i.e., $\Gamma(s_u) = \Gamma(s'_u)$,
for any two states $s_u$ and $s'_u$ reached by node $u$ after step $t$.
For a population protocol solving the majority problem, a stable configuration is {\em correct} if the output 
of all the nodes is the majority value in $X$,
in which case we say that the decision is the output majority value.
A protocol is {\em always correct} if the correct stable configuration is reached with probability 1, 
and
is {\em w.h.p. correct} if the correct stable configuration is reached with probability polynomially close to 1 (i.e., $1-n^{-c}$ for some constant $c> 0$).
For protocols which are w.h.p. correct there is a chance that
a stable state is not reached, or that an incorrect stable state may~be~reached.

The {\em stabilization time} of a population protocol is the number of steps until a correct stable 
configuration
is reached. A {\em parallel time step} corresponds to $n$ 
consecutive execution steps.
The {\em time complexity} is, by convention, the stabilization time measured in parallel steps.

\subsection{Byzantine Adversary Models}

\noindent
{\bf Adversary and faulty nodes.}
The {\em Byzantine adversary} corrupts a set of nodes $F \subseteq V$, where $|F| = f$, which become {\em faulty} and from that step, their states are fully controlled by the adversary. The remaining nodes $V \setminus F$ remain {\em correct (honest)}.
The adversary determines statically (i.e., before the execution starts, \dk{but after seeing the nodes' inputs}) or dynamically which nodes to corrupt by monitoring some properties of the execution; the adversary with the latter capability is called {\em Dynamic}, otherwise is called {\em Static}.

\noindent
{\bf Observation of the execution by the adversary.}
The adversary knows the algorithm from the very beginning of the execution, \dk{as well as the inputs of the nodes.} 
\cbr{Because the scheduler is randomized,
the adversary
can only} observe the consequences of random choices to state changes allowed by its power.
The adversary that can observe all current states of nodes is called {\em Full}, but if it could observe only the 
results 
of current interactions, \cbr{i.e., the state of a node after an exchange}
 -- it is called {\em Weak}.
Additionally, if the weak adversary can only observe which nodes participate in a current interaction, it is called {\em Weak Content-Oblivious}.
\dk{In particular, because of the randomness involved in the execution, such an adversary may not be able to predict current states of honest nodes, especially if both interacting nodes are honest.}
While honest nodes \dk{may not be able to} identify themselves or other nodes,
\dk{in the sense that the considered size of local memory is too small to store unique ids,}
we assume that the adversary and the faulty nodes have the ability 
to identify themselves~and~other~nodes.

\dk{
\noindent
{\bf Corruption of nodes.}
Static adversary has to decide which nodes are faulty in the beginning of the execution, knowing only the algorithm and nodes' inputs. Dynamic adversary, on the other hand, could {\em corrupt} a node (i.e., decide that the node becomes faulty) at any time, based on the initial knowledge of the algorithm, nodes' inputs and on the history of observations made so far (see the previous paragraph).
}

\noindent
{\bf Adversarial actions through states of faulty nodes.}
\dk{The only way for faulty nodes to affect states of honest nodes is via exchanges.}
In an exchange, any of the two chosen nodes may be faulty. 
A faulty node may simulate any honest node's state -- 
thus, although the outcome of transition function $\Delta$
is meaningless for a faulty node taking part in the exchange, 
it
could affect the resulting state of the other honest node.
\cbr{Moreover,
a faulty node has the ability to set its state based on the state of the honest node.} 
Hence,
with faulty nodes the protocol may not reach a stable correct~configuration.
\dk{Note that the exchange function is trustworthy -- a faulty node cannot arbitrarily change the state of a paired honest node, only by altering its part of the input to the function which (among others) governs the change of the state of the honest node. This is a common assumption in Byzantine-tolerant literature -- otherwise, if the faulty nodes could do such arbitrary alteration, it would mean full control over honest nodes, and thus they could not be called honest any more.}

\noindent
{\bf Basic model refined by the presence of the adversary.}
Given the faulty nodes, a configuration is {\em stable} if the output of each {\em honest} node remains the same thereafter. In a correct stable configuration, all honest nodes output the majority value. We do not care about the output of the faulty nodes in stable configurations. The notions of always correct, w.h.p. correct, and stabilization time, remain the same as in the basic model.

%% file: lowerbound.tex
\section{Lower Bounds}
\label{sec:lower-bounds}

Suppose that $a$ and $b$ are the number of nodes with initial preference $A$ and $B$, respectively.
We prove impossibility results if the initial difference $d=|a-b|$ is too small. We start from the strongest of the considered Byzantine adversaries, and then we modify the proof to work, in slightly weaker forms, against the weaker (dynamic) adversaries.

\begin{lemma}
\label{lem:lower-full}
For any $f \geq 1$, if $|a - b| < 2f$, for any majority algorithm (deterministic or randomized) there is an execution scenario such that the {\bf Full Static Byzantine adversary} causes the minority value to be decided with probability at least $p$, where $p$ is the lower bound on the probability that the algorithm decides correctly (on majority value) if there~are~no~failures.
\end{lemma}

\begin{proof}
Without loss of generality assume that $a > b$, hence $A$ is the majority value that will be decided if there are no faulty nodes.
The Byzantine adversary selects $f$ nodes out of the $a$ to make them faulty
(since the adversary is static, 
the selection of faulty nodes has already been made before the first exchange takes place).
The $f$ faulty nodes then behave as if they have initial preference $B$.
Then, the overall tally for $A$ is $a-f$ and the tally for $B$ is $b+f$.
Since $a - b < 2f$, we obtain $a - f < b + f$, and hence $B$ is the majority.
The algorithm cannot distinguish the faulty nodes from the honest ones that have initial preference $B$.
Thus, at the end of the algorithm the nodes decide $B$ overturning the correct outcome $A$.
This construction could be done with probability at least $p$, applied to each execution initially failure-free in which the algorithm decides correctly on majority value.
\end{proof}

\begin{lemma}
\label{lemma:lower-weak-adaptive}
For any $f$ such that $2 \leq f \leq n / 132$, if $|a - b| < 2(f - 1)$ then for any deterministic majority algorithm there is an execution scenario such that the {\bf Weak Byzantine adversary} causes the minority value to be decided with probability~$\Theta(1)$.
\end{lemma}

\begin{proof}
\dk{Consider a scenario in which $|a - b| < 2(f - 1)$. Since $f\ge 2$, we have $|a-b| \le 1 < 2(f-1)$ and hence we may focus on cases when $a\ne b$.}
Without loss of generality \dk{we may also} assume that $a > b$, hence $A$ is the majority value that will be decided if there are no faulty nodes.
Before the first exchange the Byzantine adversary chooses uniformly at random one of the two preference values $A$ or $B$.
Thus, suppose that the adversary picks the minority value $B$, \dk{which takes place} with probability $1/2$.
Then, the adversary will cause $f$ nodes to become faulty so that the chosen value $B$ is decided at the end of the algorithm \dk{-- this is a contradiction with a constant probability. Below we provide more details for this argument.}

Since $a - b > 0$,
every non-faulty node must participate in at least one exchange 
to reach a correct stable configuration.
An exchange is {\em first-dual} if it is the first exchange for both involved nodes.
The adversary observes first-dual exchanges between two nodes with initial preference $A$, 
and then it makes the two nodes faulty such that
they will behave as honest nodes with initial preference $B$.
In the next exchanges, 
the remaining nodes will have to assume that those two faulty nodes are honest 
with initial preference $B$,
since the faulty nodes have not participated in any other exchange before their conversion to faulty.

The Byzantine adversary will attempt to execute the above process $\lfloor f/2 \rfloor$ times
to convert in total $2\lfloor f/2 \rfloor \leq f$ 
nodes with initial preference $A$ 
to faulty nodes that behave as honest that have initial preference $B$.
Since $f - 1 \leq 2 \lfloor f/2 \rfloor$,
we get $a - b < 2(f - 1) \leq 4 \lfloor f/2 \rfloor$. 
Thus, $a - 2\lfloor f/2 \rfloor < b + 2\lfloor f/2 \rfloor$, 
and hence, \dk{after the adversarial conversion of $2 \lfloor f/2 \rfloor$ nodes, value} $B$ becomes the majority.
Thus, at the end of the algorithm the nodes will decide $B$
overturning the correct outcome $A$, \dk{which is the sought contradiction}.

\dk{It remains to prove a lower bound on the probability that such adversarial scenario is possible -- i.e., }
a lower bound for the probability of success of the Byzantine adversary to 
grab $\lfloor f/2 \rfloor$ first-dual exchanges where both nodes have initial preference $A$.
Let $X$ be the nodes with initial preference $A$.
Let $Q$ be the first $n/8$ exchanges in the algorithm execution. \dk{Recall here that exchanges are chosen by the random scheduler; in particular, w.l.o.g., we could look at the scheduler in the following way -- it first randomly generates a set of $n/8$ pairs of nodes, which defines set $Q$; next, it randomly orders these pairs, to determine which of the exchanges are actually first-dual.}\footnote{\dk{Such splitting of scheduler randomness into first selecting a random set of elements (of pairs, in our case) and then considering an independent random order (permutation) of these elements, has been already considered in the literature, cf.,~\cite{DBLP:conf/esa/KesselheimRTV13}}}
Let $Q' \subseteq Q$ be all the first-dual exchanges in $Q$, where both nodes are in $X$.
\dk{Note that in order to determine $Q'$, one needs the randomly selected set $Q$ and a random order of these pairs, as explained earlier.}
Since $|Q| = n/8$, we also have that $|Q'| \leq n/8$, which implies that $Q'$ involves 
$2|Q'| \le n/4$ nodes from $X$.
Since $|X| \geq n/2$ (recall that $A$ is the majority),
at least $n/4$ nodes from $X$ are not involved in any exchange in \dk{$Q'$.}

We say that the exchanges in $Q$ are successful \cbr{if} they are also in $Q'$, 
and are unsuccessful if they are in $Q \setminus Q'$. 
\dk{Consider the process of generating the ordered set of pairs $Q$ by the scheduler, one by one.}
Since at least $n/4$ nodes from $X$ are not involved in any exchange in $Q$,
\dk{when generating next pair in the ordered set $Q$}
there are at least $n/4$ 
nodes from $X$ \dk{which have not been in any previously selected pair.}
\dk{Hence, the probability that a newly generated pair in $Q$ is successful is at least $\frac {{n/4 \choose 2}} {{n^2/2}}$,
where ${n/4} \choose 2$ is the lower bound on the number of pairs that could be created from nodes that have not been yet selected and ${{n} \choose {2}}$ is the number of all possible pairs of nodes. Using the bound ${{n} \choose {2}} < n^2/2$ and assumption $n \geq 2 \cdot 132 > 8$, we further estimate the probability of generating a successful pair from above as follows:}
$$
\frac {{n/4 \choose 2}} {{n^2/2}} 
\ \ = \ \ \frac {(n/4)(n/4 - 1)} {n^2}
\ \ \geq \ \ \frac {1} {32} \ .
$$
Therefore, an exchange in $Q$ is unsuccessful with probability at most $31/32$.

Let $Y = |Q| - |Q'|$ be a random variable denoting the number of unsuccessful exchanges in $Q$.
We have $E[Y] \leq |Q| \cdot 31 / 32 = n/8 \cdot 31 / 32$.
Let $q = (32)^2 / (31 \cdot 33) = 1024/1023$.
From Markov's inequality,
\begin{eqnarray*}
Pr[Y \geq n/8 \cdot 32/33]
&=& Pr[Y \geq n/8 \cdot 31/32 \cdot q]\\
&\leq& Pr[Y \geq E[Y] \cdot q] 
\ \ \leq \ \ 1/q \ .
\end{eqnarray*}
Since $|Q| = n/8$, we get $Y = n/8 - |Q'|$. 
Hence,
\begin{eqnarray*}
Pr[|Q'| > n/8 \cdot 1/33] 
& = & Pr[Y < n/8 \cdot 32 / 33] \\
& > & 1 - 1/q 
\ \ = \ \ 1/1024 \ .
\end{eqnarray*}
Hence, with probability at least $1/1024$
the number of faulty nodes $f$ in the execution scenario can be such that
$\lfloor f/2 \rfloor \geq |Q'| > n/8 \cdot 1/33$, 
or equivalently $f \geq 2 \lfloor f / 2 \rfloor > n / 4 \cdot 1 / 33 = n / 132$.
This implies that with probability at least $1/1024$,
for any $f$ such that $2 \leq f \leq n/132$,
the Byzantine adversary successfully finds $2 \lfloor f / 2 \rfloor$
nodes to~convert~to~faulty.

Combining the probability of $1/1024$ with the probability of $1/2$ that the 
Byzantine adversary has correctly guessed the minority preference $B$,
we have that the above execution scenario can occur with probability at least $1/2 \cdot 1/1024 = \Theta(1)$.
\end{proof}

A variant of Lemma~\ref{lemma:lower-weak-adaptive}
could be also proved for a slightly weaker 
content-oblivious adversary, however for a smaller difference $d$.

\begin{lemma}
\label{lemma:lower-weak-oblivious}
For any $f$ such that $2 \leq f \leq n / 132$, if $|a - b| < f - 1$ then for any deterministic majority algorithm there is an execution scenario such that the {\bf Weak Content-Oblivious Byzantine adversary} causes the minority value to be decided with probability~$\Theta(1)$.
\end{lemma}

\begin{proof}
The proof is almost identical to the proof of Lemma~\ref{lemma:lower-weak-adaptive}, except that the adversary cannot see the state information of the exchange and cannot identify first-dual exchanges with probability $1$, as in the other proof.
However, by observing the traffic it could identify {\em first-dual-oblivious} exchanges, i.e.,  
exchanges which are first exchanges for their both participants (without guarantee that they have the same and desired initial value). Therefore, each among the corrupted $\lfloor f/2 \rfloor$ pairs may have two, one or none desired value to be flipped by corrupted nodes ($A$).
If $f \leq n/132$, the probability of getting at least one value $A$ in the first-dual-oblivious exchange is at least $1/2-2f/n \ge 2/5$, therefore the probability in the proof of Lemma~\ref{lemma:lower-weak-adaptive} needs to be adjusted by factor $2/5$ (which is still only another constant factor, and could be hidden in ``$\Theta$'' notation), while the difference $|a-b|$ will be only by half affected comparing with the case of the weak adversary, hence 
$|a-b|< f-1$.
\end{proof}

\remove{
We will prove that during $k$ exchanges the number of distinct honest nodes that have performed exchanges with the faulty nodes
is $\Omega(k f / n)$.

\begin{lemma}
During $k$ successive exchanges, the number of distinct honest nodes that have performed exchanges with faulty nodes is $\Omega(k f / n)$ with probability at least $1 - o(1)$.
\end{lemma}

\begin{proof}
Let $p=2(a-kf/n)f/n^2$.
Choose a set $B\subseteq A$ of distinct honest nodes interacting with Byzantine nodes.
The probability that at least one node in $A\setminus B$ is selected in $kf/n$ interactions is:
\[
(1-p)^{k} =
(1-2(a-kf/n)f/n^2)^{k}
\le 
\exp(-2kf(a-kf/n)/n^2)
\ .
\]
The probability that there is a set $B$ of size $kf/n\le a/2$ such that an honest element outside it is selected is at most
\[
\binom{a}{kf/n} \exp(-2kf(a-kf/n)/n^2)
\le
(ane/(kf))^{kf/n}
\exp(-2kf(a-kf/n)/n^2)
\]
\[
\le 
\exp((kf/n)\ln (ane/(kf)) - 2kf(a-kf/n)/n^2)
\]
\[
\le 
\exp((kf/n)[\ln (ane/(kf)) - 2(a-kf/n)/n])
\]
For $kf/n=a/2$ we have:
\[
\binom{a}{kf/n} \exp(-2kf(a-kf/n)/n^2)
\le
2^a/\sqrt{a} \cdot
\exp(-a^2/n)
\]
\[
\le 
\exp(a(\ln 2 - a/n))/\sqrt{a}
\]
\end{proof}
}

%% file: skeleton.tex
\cbr{
\section{Skeleton Algorithm}
\label{section:skeleton}
}
\begin{algorithm}[t]
{\footnotesize \cbr{
\caption{\sf Skeleton algorithm for node $u$}
\label{alg:skeleton}
\tcp{Initialization for node $u$}
Preference is in $value_u \in \{A,B, \bot \}$, initially set to $A$ or $B$\;
$decision_u \gets \bot$; 
$saved_u \gets \bot$; 
$C_u \gets -1$; $phase_u \gets 0$\;

\BlankLine
\tcp{Actions on exchange $e_i$ involving $u$}
\ForEach{exchange $e_i = \{ u,v \}$}{
\tcp{$D = \Theta(\log^2 n)$ is the phase duration}
\tcp{$maxPhases = \Theta(\log n)$}

\If{$phase_u \leq maxPhases$}{
$C_u \gets (C_u + 1) \mod D$; \tcp{increment counter}
\If{$C_u = 0$}{
$phase_u$++; \tcp{enter new phase}
$saved_u \gets value_u$; \tcp{save current value}
}
}
\If{$phase_u > maxPhases$}{Skip next steps of for-loop\;}

\BlankLine
\tcp{From $C_u$ determine the current phase kind} 
\BlankLine
\tcp{Cancellation Phase}
\tcp{\cbr{Phase repeats $\gamma = \Theta(1)$ times}}
\If{$u$ is in cancellation phase and $phase_u = phase_v$}
{Cancellation Phase Actions\;
}
\BlankLine

\tcp{Resolution Phase}
\If{$u$ is in the resolution phase and $decision_u = \bot$}
{
\If{$u$ is in first subphase}
{$probe_u^A \gets 0$; $probe_u^B \gets 0$\;}
\If{$u$ is in the second subphase}
{$u$ invokes Algorithm \ref{alg:decision-check}: {\sf Decision-Check}$(e_i)$\;}
}

\BlankLine
\tcp{Duplication Phase}
\If{$u$ is in duplication phase and $phase_u = phase_v$}
{Duplication Phase Actions\;
}
}
}
}
\end{algorithm}

\begin{algorithm}[t]
{\footnotesize
\caption{{\sf Decision-Check}$(e_i)$}
\label{alg:decision-check}
\tcp{$e_i = \{u,v\}$ and algorithm is invoked by $u$}

\If{$value_v = A$} {$probe_u^A$++; \tcp{$u$ sampled A}}
\If{$value_v = B$} {$probe_u^B$++; \tcp{$u$ sampled B}}
\tcp{$\psi = O(\log n)$ is number of samples}
\tcp{$\sigma_1 = O(\log n), \sigma_2 = O(\log n)$ are threshold values}
\If{$e_i$ is the $\psi$th exchange in the second subphase of $u$}{
\If{$probe_u^B \leq \sigma_1 \wedge probe_u^A \geq \sigma_2$ } 
{$decision_u \gets A$; \tcp{$u$ decides $A$}}
\If{$probe_u^A \leq \sigma_1 \wedge probe_u^B \geq \sigma_2$ } 
{$decision_u \gets B$; \tcp{$u$ decides $B$}}
}
}
\end{algorithm}

\cbr{
Algorithm \ref{alg:skeleton} is the skeleton psudocode for node $u \in V$.  
This skeleton is common for Algorithms \ref{alg:acpd} (\acpd) and Algorithm \ref{alg:scfd} (\scfd).
Each exchange that involves node $u$ triggers the following
steps: (i) first update a local counter and a phase counter and (ii) then 
proceed to execute one of the the cancellation, resolution, and duplication phases.

In the algorithms we assume that each node 
maintains its state in a set of local variables.
In an exchange between an honest and faulty node, 
if the faulty node does not provide access to the value of its local 
variable, then the honest node assumes an arbitrary value for that variable.
}

\vspace*{1ex}
\noindent
\cbr{{\bf Local counter:}}
Suppose that the randomized scheduler generates a sequence of exchanges $e_1, e_2, \ldots$,
where $e_i \in \{\{u,v\}:u,v \in V \}$.
\cbr{The {\em local counter} $C_u$ is a variable of node $u \in V$ (saved in the state of $u$)}
which is 
incremented in each exchange that $u$ participates in,
so that at the $r$th exchange by the scheduler, $C_u = |\{e_i: u \in e_i \wedge i \leq r\}|$.
The local counter is incremented modulo $D = \Theta(\log^2 n)$,
which means that it takes $\Theta(\log \log n)$ bits to save the counter locally.
We use the local counters to synchronize the actions of the nodes.
In Lemma \ref{lemma:drift} we show that within the first $\Theta(n \log^3 n)$ system exchanges,
with high probability any two local counters differ by at most $O(\log^2 n)$.

\noindent
\cbr{{\bf Phase counter:}}
The $O(\log^2 n)$ bound on the maximum local counter difference 
enables us to divide the execution of the algorithm 
into phases such that the duration of a phase is $D = \Theta(\log^2 n)$ exchanges.
Each node has its own perception of a phase based on its local counter.
A phase for a node $u$ consists of $D$ consecutive exchanges
as measured by the local counter of $u$.
Node $u$ keeps track of the phase count that it is currently in
with the use of a variable $phase_u$ that increments 
each time that the local counter $C_u$ becomes $0$.
The $C_u$ value is also used to determine the current phase kind 
(cancellation, resolution, duplication).

Because local counters may differ,
nodes could be in different phases.
The value $D$ is chosen so that any two nodes are at any moment
in two consecutive phases of the algorithm, with high probability.
Moreover, in order to have all the nodes at some point in the same phase concurrently,
we divide each phase into three equal length subphases, 
each consisting of $D/3$ exchanges.
The value $D$ is chosen so that if one node $u$ is in the middle subphase of \cb{the} phase, 
then the other nodes are with high probability within the same phase,
that is, $phase_u = phase_v$
(but $v$ is in any subphase~of~$phase_v$).
\cbr{
In particular, $D = 6 \lceil \zeta \rceil$, where $\zeta = \sqrt{12 c} \ln^2 n$ (Lemma~\ref{lemma:drift}),
for some appropriate constant $c$
which is the factor in the asymptotic expression
for the number of exchanges $O(n \log^3 n)$ as
inferred from the analysis of Theorems \ref{theorem:algo1} and \ref{theorem:algo2}.}

\noindent
\cbr{{\bf Decision value:}}
Each node $u$ starts with an initial value which is either $A$ or $B$, 
stored in local variable $value_u$.
During the execution of the algorithm the contents of $value_u$ change according to the 
current preference of the node.
Moreover,
during cancellation exchanges the content of the $value_u$ may become empty, denoted as $\bot$.
Therefore, throughout the algorithm execution, $value_u \in \{A, B, \bot\} $.
Just before node $u$ enters a new phase,
it saves the current value in the local variable $saved_u$.

Node $u$ has also a local variable $decision_u$ to store the majority value that it decides.
Initially $decision_u = \bot$ and eventually it sets the value of $decision_u$ to either $A$ or $B$.
When $u$ sets $decision_u \neq \bot$ we also say that the node $u$ has ``decided'', 
meaning that it has reached a stable state for the output 
where it will no longer change its decision further.
\cbr{A node decides in a resolution phase, but not all nodes may decide in the same 
resolution phase.
A decided node is still participating in subsequent 
exchanges and executes the other phases,
but is not updating its decision.
}

\noindent
\cbr{
{\bf Phase execution:}
There are three types of phases:}
\begin{itemize}
\item 
\cbr{
{\em Cancellation phase:} non-empty nodes with value $A$ or $B$, attempt to match their value with an opposite value, where the opposite of $A$ is $B$ and vice-versa. This results to empty nodes with value $\bot$.
}
\item 
\cbr{
{\em Duplication phase:} non-empty nodes with value $A$ or $B$ duplicate their values into empty nodes. 
}

\item 
\cbr{
{\em Resolution phase:} nodes attempt to make a decision on the majority value.
The resolution phase is implemented by invoking Algorithm {\sf Decision-Check}$(e_i)$ 
 (Algorithm \ref{alg:decision-check}).}
Node $u$ checks in the first $\psi$ exchanges of its second resolution subphase the number of times that it observes the values $A$ and $B$ in other nodes during the exchanges (a form of sampling values of other nodes). These numbers are recorded in local variables $probe_u^A$ and $probe_u^B$ (see Algorithm \ref{alg:decision-check}). If $probe_u^A \geq\sigma_2$ while $probe_u^B \leq \sigma_1$, then node $u$ decides $A$. If $probe_u^B \geq\sigma_2$ while $probe_u^A \leq \sigma_1$ then node $u$ decides $B$. Otherwise, no decision is made. The decision is recorded in the local variable $decision_u$ that gets a value in $\{A,B\}$. In the analysis we give the precise expressions for $\psi$, $\sigma_1$, and $\sigma_2$. 
\cbr{If a decision is made then 
it is for the majority value.}
\end{itemize}

\cbr{The actual order of phase execution is first cancellation, 
then resolution, followed by duplication.}
A node starts in the cancellation phase and it actually repeats the cancellation phase a constant $\gamma = \Theta(1)$ consecutive times before proceeding to a different phase kind.
Then it continues to the resolution phase once, and finally to the duplication phase also once.
This cycle repeats $O(\log n)$ times where each cycle involves $\gamma + 2 = \Theta(1)$ phases.
The majority decision 
can be made during any of the cycles.
The final decision value for $u$ appears in $decision_u$.

There is an unlikely (with low probability) 
scenario that some node $u$ does not reach a decision.
In that case, the $phase_u$ variable \cbr{exceeds} an upper threshold value $maxPhases = \Theta(\log n)$,
for an appropriate choice of constant factor in the asymptotic formula
as determined in the analysis,
while $decision_u = \bot$.
The bound of $maxPhases$ prevents the variable $phase_u$ from increasing \cbr{indefinitely}.

\cbr{We continue with a basic analysis of Algorithm~ \ref{alg:skeleton}
which is common in both instances in Algorithms \ref{alg:acpd} and \ref{alg:scfd}.}

\cbr{
\subsection{Local Counter Analysis}
We prove a basic result for the local counters of Algorithm~ \ref{alg:skeleton}.
We use the following version of the Chernoff bound.
}

\begin{lemma}[Chernoff bound]
	\label{lemma:chernoff}
	Let $X_1,X_2,\cdots,X_{m}$ be independent Poisson trials 
	such that, for $1\leq i\leq m$, 
	$X_i \in \{0,1\}$.
	Then, for $X=\sum_{i=1}^{m}X_i$, $\mu=E[X]$,
	\cb{ and any $\delta > 1$,
	$\Pr[X \geq (1+\delta) \mu] \leq e^{-\frac{\delta \mu} 2}$,
	and any $0 \leq \delta \leq 1$,
 $\Pr[X \geq (1+\delta) \mu] \leq e^{-\frac{\delta^2 \mu} 3}$
    and $\Pr[X \leq (1-\delta) \mu] \leq e^{-\frac{\delta^2 \mu} 2}$.}
\end{lemma}

We continue with bounding the difference between the local counters of the nodes
which is important for synchronizing the phases of the nodes.
In the next lemma, the constant $c$ is used adjust the length of the phases according to the needs of the algorithms.
In many explicit formulas, we use the natural logarithm $\ln x$ that has base $e$ while
in asymptotic \cbr{notations} we use $\log x$.

\begin{lemma}
\label{lemma:drift}
\cbr{For $n \geq e^2$,} 
at the $r$th exchange issued by the scheduler, 
where $r \leq c n \ln^3 n$
\cbr{and constant $c \geq 1$},
for all pairs of nodes $u$ and $v$
it holds that $|C_u - C_v| < 2 \zeta$,
where $\zeta = \sqrt{12 c} \ln^2 n$,
with probability at least $1 - 2/n$.
\end{lemma}

\begin{proof}
For $n$ nodes 
there are ${{n} \choose {2}}$ possible exchange pairs,
out of which a node belongs in $n-1$ possible exchange pairs.
Thus, a node is included in a randomly chosen exchange 
with probability $p = (n-1) / {n \choose 2} = 2 / n$.

For a node $u$,
let $X^{(i)}_u$ be the random variable such that $X^{(i)}_u = 1$ if $u$ is picked at the $i$th exchange,
and otherwise, $X^{(i)}_u = 0$.
The $X^{(i)}_u$ (for any given $u$) 
are independent Bernoulli trials
with $\Pr[X^{(i)}_u =  1] = p$.

The number of times that node $u$ was picked during $r$ exchanges 
is $C_u = \sum_{i=1}^{r} X^{(i)}_u$. 
We have that the expected value of $C_u$ is $\mu = E[C_u] = r p = 2 r / n$.
Let $\delta = \zeta / \mu$ (note that $\delta \leq 1$ \cbr{for $n \geq e^2$}).
From Lemma \ref{lemma:chernoff},
\begin{eqnarray*}
\Pr[C_u \geq \mu + \zeta]
& = & \Pr[C_u \geq (1 + \zeta/\mu)\mu] \\
& = & \Pr[C_u \geq (1 + \delta)\mu] \\
& \leq & e^{-\frac{\delta^2 \mu} 3}
\ \  = \ \  e^{-\frac{\zeta^2 } {3 \mu}}
\ \ \leq \ \ e^{-\frac{12 c \ln^4 n} {6 c \ln^3 n}} \\
& = & e^{-2\ln n} 
\ \ = \ \ 1/n^2 \ .
\end{eqnarray*}
Similarly,
\begin{eqnarray*}
\Pr[C_u \leq \mu - \zeta]
& = & \Pr[C_u \leq (1 - \zeta/\mu)\mu]\\
& = & \Pr[C_u \leq (1 - \delta)\mu]\\
& \leq & e^{-\frac{\delta^2 \mu} 2}
\ \ = \ \ e^{-\frac{\zeta^2 } {2 \mu}}
\ \ \leq \ \ e^{-\frac{12 c \ln^4 n} {4 c \ln^3 n}} \\
& = & e^{-3\ln n}
\ \ = \ \ 1/n^3 \ .
\end{eqnarray*}
Therefore,
$\Pr[|\mu - C_u| \geq \zeta] \leq 1 / n^2 + 1/n^3 < 2/n^2$.
Considering all the $n$ nodes and using the union bound,
any of these nodes has its local counter to differ by $\zeta$ or more from $\mu = 2r/n$ 
with probability at most $n \cdot 2/n^2 = 2/n$.
Thus, with probability at least $1 - 2/n$,
for all nodes $u$ and $v$, $|C_u - C_v| < 2 \zeta$,
as needed.
\end{proof}

\cbr{\subsection{Resolution Phase Analysis}}

Next, we consider the resolution phase.
\cbr{We analyze Algorithm \ref{alg:decision-check}.}
Suppose that $f \leq n/c_f$, for some constant $c_f$.
In the resolution phase a node observes the first $\psi = c_\psi \ln n$ exchanges in the second sub-phase, for some constant $c_\psi$.
During the $\psi$ exchanges a node counts the number of exchanges with nodes that have value $A$ and with nodes that have value $B$.
Let $\sigma_1 = c_{\sigma_1} \ln n$ and $\sigma_2 = c_{\sigma_2} \ln n$, for some constants $c_{\sigma_1}, c_{\sigma_2} \leq c_\psi $. A node $u$ decides value $X \in \{A,B\}$ if $X$ appears in at least $\sigma_2$ exchanges (out $\psi$), while at most $\sigma_1$ exchanges are with nodes that have the other opposite value $Y \in \{A,B\} \setminus \{ X \} $.
Otherwise, there is no decision made in the resolution phase by node $u$.
In the following two lemmas we \cbr{use as constants} $\xi, \xi_1, \xi_2 \geq 1$, where $\xi_1 \geq \xi \geq \xi_2$,  such that:
\begin{align}
 \label{eqn:ksi1}
 2 c_{\sigma_1} \leq & \left (\frac {1} {\xi} + \frac {1} {c_f}\right ) c_\psi \leq  \frac {c_{\sigma_2}} {2} \ , \\
 \label{eqn:ksi2}
 & \left (\frac {1} {\xi_1} + \frac {1} {c_f} \right) c_\psi \leq \frac {c_{\sigma_1}} {2} \ ,\\
 \label{eqn:ksi3}
 2 c_{\sigma_2} \leq & \left (\frac {1} {\xi_2} + \frac{1} {c_f} \right) c_\psi \ .
\end{align}
We can assume the following example constants that satisfy Equations \ref{eqn:ksi1}, \ref{eqn:ksi2}, \ref{eqn:ksi3}: $c_f = 256, \xi = 32, \xi_1 = 256, \xi_2 = 4, c_{\sigma_1} = 12, c_{\sigma_2} = 8 c_{\sigma_1}, c_\psi = 64 c_{\sigma_1}$.
The next sequence of lemmas will lead to Lemma \ref{lemma:resolution} for the resolution phase.

\begin{lemma}
\label{lemma:resolution-min1}
In the resolution phase,
if a node observes value $X$ in at most $\sigma_1$ exchanges,
then  $X$ is the current value of at most $n/\xi$ honest nodes with probability at least $1 - 1/n$.
\end{lemma}

\begin{proof}
Consider the $\psi$ exchanges of honest node $u$ during the second sub-phase of the resolution phase. 
Suppose that the number of honest nodes that have the value $X$ at the beginning of the phase is $x > n/\xi$.
The probability that an exchange $\{ u,v\}$ occurs where $v$ is either a honest node with value $X$ 
or a faulty node pretending to have value $X$ is at least $(x + f) / n > 1/\xi + 1/c_f $.
Thus, 
using Equation \ref{eqn:ksi1},
for $(1/\xi + 1/c_f) \cdot c_\psi \geq 2 c_{\sigma_1}$,
$c_{\sigma_1} = 12$, the expected count $X$ of the majority value by node $u$ is at least %
\cbr{
\begin{eqnarray*}
E[X] 
& > &  (1/\xi + 1/c_f) \cdot \psi
 =  (1/\xi + 1/c_f) \cdot c_\psi \ln n \\
& \geq & 2 c_{\sigma_1} \ln n
 =  24 \ln n \ .
\end{eqnarray*}
}
From Lemma \ref{lemma:chernoff},
for $\delta = 1/2$ and $\mu = 24\ln n$, 
the probability that the count $X$ is at most $\sigma_1 = c_{\sigma_1} \ln n = 12\ln n = \mu/2 = (1 -\delta) \mu$,
is at most $e^{-\delta^2 \mu / 2} = e^{- 3 \ln n} \leq n^{-2}$.

Thus, with probability at most $n \cdot n^{-2} = n^{-1}$, 
the count for value $X$ is at most $\sigma_1$ in any of the $n$ nodes.
Hence, with probability at least $1 - 1/n$ the count for value $X$ is more than $\sigma_1$ in all the nodes.
Thus, if a node observes that the count for $X$ is at most $\sigma_1$ then with probability at least $1 - 1/n$ 
it holds that $x \leq n/\xi$.
\end{proof}

\begin{lemma}
\label{lemma:resolution-min2}
In the resolution phase,
if $X$ is the current value of at most $n/\xi_1$ honest nodes
then every node observes value $X$ in at most $\sigma_1$ exchanges,
with probability \cbr{at least $1-1/n$}.
\end{lemma}

\begin{proof}
Consider the $\psi$ exchanges of honest node $u$ during the second sub-phase of the resolution phase. 
Suppose that the number of honest nodes that have the value $X$ at the beginning of the phase is $x \leq n/\xi_1$.
The probability that an exchange $\{ u,v\}$ occurs where $v$ is either a honest node with value $X$ 
or a faulty node pretending to have value $X$ is at most $(x + f) / n \leq 1/\xi_1 + 1/c_f $.
Thus, 
using Equation \ref{eqn:ksi2}, 
for $(1/\xi_1 + 1/c_f) \cdot c_\psi \leq c_{\sigma_1}/2$,
$c_{\sigma_1} = 12$, the expected count $X$ of the majority value by node $u$ is at most
\cbr{
\begin{eqnarray*}
E[X] 
& \leq &  (1/\xi_1 + 1/c_f) \cdot \psi
 =  (1/\xi_1 + 1/c_f) \cdot c_\psi \ln n\\
& \leq & (c_{\sigma_1} \ln n)/2
 =  \sigma_1/2 = 6 \ln n \ .
\end{eqnarray*}
}
From Lemma \ref{lemma:chernoff},
for $\delta = 1$ and $\mu = 6\ln n$, 
the probability that the count $X$ is at least $\sigma_1 = c_{\sigma_1} \ln n = 12\ln n = 2\mu = (1 +\delta) \mu$,
is at most $e^{-\delta^2 \mu / 3} = e^{- 2 \ln n} = n^{-2}$.

Thus, with probability at most $n \cdot n^{-2} = n^{-1}$, 
the count for value $X$ is at least $\sigma_1$ in any of the $n$ nodes.
Hence, with probability at least $1 - 1/n$ the count for value $X$ is at most $\sigma_1$ in all the nodes.
\end{proof}

\begin{lemma}
\label{lemma:resolution-max1}
In the resolution phase,
if a node observes value $Y$ in at least $\sigma_2$ exchanges, 
then  $Y$ is the current value of more than $n/\xi$ honest nodes with probability at least $1 - 1/n$.
\end{lemma}

\begin{proof}
Consider the $\psi$ exchanges of honest node $u$ during the second sub-phase of the resolution phase. 
Suppose that the number of honest nodes that have the value $Y$ at the beginning of the phase is $y \leq n/\xi$.
The probability that an exchange $\{ u,v\}$ occurs where $v$ is either a honest node with value $Y$ 
or a faulty node pretending to have value $Y$ is at most $(y + f) / n \leq 1/\xi + 1/c_f $.
Thus, 
using Equation \ref{eqn:ksi1},
for $(1/\xi + 1/c_f) \cdot c_\psi \leq c_{\sigma_2}$/2,
$c_{\sigma_2} = 96$ (recall that $c_{\sigma_2} = 8 c_{\sigma_1} = 8 \cdot 12 = 96$), the expected count $Y$ of the majority value by node $u$ is at most 
\cbr{
\begin{eqnarray*}
E[Y] 
& \leq & (1/\xi + 1/c_f) \cdot \psi
 =  (1/\xi + 1/c_f) \cdot c_\psi \ln n \\
& \leq & (c_{\sigma_2} \ln n)/2
 =  48 \ln n \ .
\end{eqnarray*}
}
From Lemma \ref{lemma:chernoff},
for $\delta = 1$ and $\mu = 48\ln n$, 
the probability that the count $Y$ is at least $\sigma_2 = c_{\sigma_2} \ln n = 96\ln n = 2 \mu = (1 +\delta) \mu$,
is at most $e^{-\delta^2 \mu / 3} = e^{- 48/3 \cdot \ln n}\leq e^{- 2 \ln n} = n^{-2}$.

Thus, with probability at most $n \cdot n^{-2} = n^{-1}$, 
the count for value $Y$ is at least $\sigma_2$ in any of the $n$ nodes.
Hence, with probability at least $1 - 1/n$ the count for value $Y$ is less than $\sigma_2$ in all the nodes.
Thus, if a node observes that the count for $Y$ is at least $\sigma_2$ then with probability at least $1 - 1/n$ 
it holds that $y > n/\xi$.
\end{proof}

\begin{lemma}
\label{lemma:resolution-max2}
In the resolution phase,
if $Y$ is the current value of at least $n /\xi_2$ honest nodes
then every node observes value $Y$ in at least $\sigma_2$ exchanges,
with probability \cbr{at least $1-1/n$}.
\end{lemma}

\begin{proof}
Consider the $\psi$ exchanges of honest node $u$ during the second sub-phase of the resolution phase. 
Suppose that the number of honest nodes that have the value $Y$ at the beginning of the phase is $y \geq n /\xi_2$.
The probability that an exchange $\{ u,v\}$ occurs where $v$ is either a honest node with value $Y$ 
or a faulty node pretending to have value $Y$ is at least $(y + f) / n \geq 1/\xi_2 + 1/c_f $.
Thus, 
using Equation \ref{eqn:ksi3},
for $(1 / \xi_2 + 1/c_f) \cdot c_\psi \geq 2 c_{\sigma_2}$,
$c_{\sigma_2} = 96$ (recall that $c_{\sigma_2} = 8 c_{\sigma_1} = 8 \cdot 12 = 96$), the expected count $Y$ of the majority value by node $u$ is at least 
\cbr{
\begin{eqnarray*}
E[Y] 
& \geq & (1/\xi_2 + 1/c_f) \cdot \psi
 =  (1/\xi_2 + 1/c_f) \cdot c_\psi \ln n\\ 
& \geq & 2 c_{\sigma_2} \ln n
 = 2 \cdot 96 \ln n \ .
\end{eqnarray*}
}
From Lemma \ref{lemma:chernoff},
for $\delta = 1/2$ and $\mu = 2 \cdot 96\ln n$, 
the probability that the count $Y$ is at most $\sigma_2 = c_{\sigma_2} \ln n = 96 \ln n = \mu/2 = (1 -\delta) \mu$,
is at most $e^{-\delta^2 \mu / 2} = e^{- 2 \cdot 96 /8 \cdot \ln n}\leq e^{- 2 \ln n} = n^{-2}$.

Thus, with probability at most $n \cdot n^{-2} = n^{-1}$, 
the count for value $Y$ is at most $\sigma_2$ in any of the $n$ nodes.
Hence, with probability at least $1 - 1/n$ the count for value $Y$ is at least $\sigma_2$ in all the nodes.
\end{proof}

We combine the above results about the resolution phase 
(Lemmas \ref{lemma:resolution-min1}, \ref{lemma:resolution-min2},
 \ref{lemma:resolution-max1}, and \ref{lemma:resolution-max2})
 into the following lemma.

\cbr{
\begin{lemma}
\label{lemma:resolution}
With probability at least $1-4/n$,
the following properties hold for a resolution phase:
\begin{itemize}
\item[i.]
All honest nodes that make a decision in the same resolution phase,
correctly decide on the majority value.

\item[ii.]
If at least $n/\xi_2$ non-empty nodes have majority value $Y \in \{A,B\}$, 
and at most $n/\xi_1$ non-empty honest nodes have minority value $X \in \{A,B\}, X \neq Y$, then all honest undecided nodes will decide the majority value.
\end{itemize}
\end{lemma}
}

\begin{proof}
\cbr{
Consider a resolution phase.
From Algorithm \ref{alg:decision-check}, 
in a resolution phase,
a node decides the value $Y \in \{A,B\}$ that is observed in at least $\sigma_2$ exchanges,
only if the other value in $X \in \{A,B\}, X \neq Y$, is observed in at most $\sigma_1 $ exchanges.

Property (i) is a consequence of Lemmas \ref{lemma:resolution-min1} and \ref{lemma:resolution-max1}.
That is, the threshold $n/\xi$ separates the majority $Y$ and minority $X$ values,
so that each honest node $u$ observes value $X$ in at most $\sigma_1$ exchanges
and value $Y$ in at \cbr{least} $\sigma_2$ exchanges,
with probability at least $1 - 2/n$.
Thus, value $X$ is the current minority and value $Y$ is the current majority among all the non-empty honest nodes that decide during the phase.

Property (ii) is a consequence of Lemmas \ref{lemma:resolution-min2} and \ref{lemma:resolution-max2}.
That is, if at most $n/\xi_1$ honest nodes have value $X$ and at least $n/\xi_2$ nodes have value $Y$, then each honest node will observe $X$ in at most $\sigma_1$ exchanges and value $Y$ in at least $\sigma_2$ exchanges,
with probability at least $1 - 2/n$.}
\end{proof}

%% file: algolog3.tex
\section{Algorithm \acpd}
\label{sec:acpd}


\begin{algorithm}[t]
{\footnotesize \cbr{
\caption{\sf Asymmetric-C-Partial-D}
\label{alg:acpd}



\tcp{Cancellation Phase Actions}
\If{$($$e_i$ is the first exchange in the second subphase of $u$$)$ and
    $((value_u = A \wedge saved_v = B) \vee (value_u = B \wedge saved_v = A))$} 
    {$value_u \gets \bot$; \tcp{$u$ cancels its own value unilaterally (asymmetrically)}}
\BlankLine
\tcp{Duplication Phase Actions}
\If{$($$e_i$ is the first exchange in the second subphase of $u$$)$ and $(value_u = \bot)$ and $(saved_v \in \{A,B\}) $ \linebreak}
{$value_u \gets saved_v$; \tcp{$u$ duplicates $v$'s value}} 
}
}
\end{algorithm}

\cbr{
Algorithm \acpd\ 
(Algorithm~\ref{alg:acpd}),
follows the structure of Algorithm~\ref{alg:skeleton}.
We fill-in the omitted action steps for the cancellation and duplication phases.
}

\cbr{
\noindent
{\bf\em Cancellation phase:} 
The cancellation phase repeats $\gamma = 1024$ times 
(as determined in the proof of Theorem~\ref{theorem:algo1}).
We describe the actions of one such phase (out of $\gamma$).}
Suppose that exchange $e_i = \{u,v\}$ is the first exchange for $u$ in the second cancellation subphase while $v$ also is in the cancellation phase (in any subphase). Then, $u$ compares $value_u$ with $saved_v$. If $value_u$ is non empty and opposite than $saved_v$ then $u$ cancels its value and $u$ becomes an empty node by setting $value_u \gets \bot$. Note that the exchange is not affecting the local state of node $v$, that is, the cancellation is {\em asymmetric}. Node $u$ will not attempt to make any further changes in its local state in the current phase. If node $u$ was not successful in cancelling its value (i.e. $saved_v$ is empty or not opposite) then node $u$ will not try to cancel again in this phase and will wait to attempt cancellation again in a subsequent phase \cbr{(in the sequence of $\gamma$ phases)}. Thus, node $u$ has a single chance in one exchange of the current phase to cancel its value and become empty. 

\noindent
{\bf\em Duplication phase:}
Suppose that $u$ participates in an exchange $e_i = \{u,v\}$ such that $value_u = \bot$ ($u$ is empty). If $e_i$ is the first exchange for $u$ in the second duplication subphase, then $u$ will attempt to duplicate (copy) the value of $v$
by setting $value_u \gets saved_v$. This exchange does not alter the local state of $v$. This is the only attempt of $u$ to modify its empty value in the current phase. If the copying is unsuccessful (because $saved_v = \bot$) then node $u$ will not try again at the current phase. Thus, $u$ has a single chance in one exchange of the phase to become non-empty.

\subsection{Analysis of Algorithm \acpd}

Consider the cancellation phase.
We say that the cancellation phase starts when for the first time
some node participates in an exchange in the second subphase of the phase,
and the cancellation phase ends when the last node takes step
in the second subphase of the phase.
We also use the use the same notions of phase start and phase end 
for the duplication and resolution phases as well.
Without loss of generality assume that $A$ is the majority value;
the proofs below are symmetric if $B$ was the majority.
Let $a$ and $b$ denote the number of honest nodes with values $A$ and $B$, respectively,
just before the cancellation phase starts.
Let $a'$ and $b'$ denote the number of honest nodes with values $A$ and $B$, respectively,
just after the cancellation phase ends.

\begin{lemma}
\label{lemma:cancellation}
In the cancellation phase, if $|a - b| \geq f + 4 \sqrt{n \ln n}$
\cb{and $ab \geq 3 n \ln n$}
then with probability at least $1 - \cb{4/n}$ 
it holds that the majority value stays the same at the end of cancellation phase such~that
\begin{itemize}
\item
$a' \leq a - ab/n + \sqrt{2 n \ln n}$,
\item
$b' \leq b - ab/n + \sqrt{2 n \ln n}$, and 
\item
$|a'-b'| \geq |a - b| - f - 4 \sqrt{n \ln n}$.
\end{itemize}
\end{lemma}

\begin{proof}
When a node $u$ enters the second sub-phase of the cancellation phase, 
in the first exchange $\{u,v\}$ in the sub-phase it decides whether to cancel its own value.
If $v$ had the opposite value than $u$ when $v$ entered the cancellation phase,
then $u$ will cancel its own value and it will become an empty node.
Note that $v$ may no longer hold a value at the moment of the exchange with $u$,
but this does not matter since we only care about the value that $v$ had at the 
beginning of the phase.
If $v$ does not have the specific opposite value,
then $u$ will not cancel its own value in~this~phase.

First consider the case that $u$ has the value $A$ at the beginning of the cancellation phase.
The probability that $v$ has the value $B$ is $b/n$.
Moreover, if $v$ is faulty then $v$ may pretend that it has the value $B$ causing cancellation for $u$.
This occurs with probability $f/n$.
Overall, with probability  $p$, 
where $b/n \leq p \leq (b+f)/n$, 
node $u$ will successfully cancel its value.

Let $a_c$ be the number of nodes with value $A$ that get cancelled during the phase. 
\cb{We have that the expected value is $\mu = E[a_c] = ap$. 
Let $\delta = \sqrt{3 \ln n} / \sqrt {\mu}$. 
Since $\mu = ap \geq ab / n \geq 3 \ln n$, 
we get 
$0 < \delta \leq 1$.
Hence, from Lemma \ref{lemma:chernoff}, 
since $af/n \leq f$ and $a(b + f) \leq n^2$},
we get that the number of cancellations for $A$ valued nodes 
is more than
\begin{eqnarray*}
\zeta_1
& = & \mu + \sqrt{3 \mu \ln n}\\
& \cb{\leq} &  a(b+f)/n + \sqrt{3 a(b+f) (\ln n)/n}\\
& \leq & ab/n + f + \sqrt{3 n \ln n}
\end{eqnarray*}
with probability at most $e^{-\delta^2 \mu / 3} = e^{-\ln n} = 1/n$.
\cb{Similarly,
for $\delta = \sqrt{2 \ln n} / \sqrt{\mu}$,}
we get that the number of cancellations for $A$ valued nodes~is~less~than
\begin{eqnarray*}
\zeta_2 
& = & \mu - \sqrt{2 \mu \ln n} \\
& \cb{\geq} & ab/n - \sqrt{2 a(b +f) (\ln n)/n} \\
& \geq & ab/n - \sqrt{2 n \ln n}
\end{eqnarray*}
with probability at most $e^{-\delta^2 \mu / 2} \leq e^{-\ln n} = 1/n$.
Therefore, with probability at least $1 - 2/n$,
the number of cancellations $a_c$ for nodes with value $A$ is bounded 
as $\zeta_2 \leq a_c \leq \zeta_1$~which~gives
$$
ab/n - \sqrt{2 n \ln n} \leq a_c \leq a b/n + f + \sqrt{3 n \ln n} \ .
$$
\cbr{
The number of cancellations $b_c$ for nodes with value $B$ 
is similarly bounded
with probability at least $1-2/n$
}
$$
ab/n - \sqrt{2 n \ln n} \leq b_c \leq ab/n + f + \sqrt{3 n \ln n} \ .
$$
Without loss of generality suppose that $a \geq b$.
Since $a_c = a - a'$ and $b_c = b - b'$,
we get
\begin{eqnarray*}
a'-b'
& = & (a - a_c) - (b - b_c)
 =   a - b - a_c + b_c \\
& \geq & a - b - (ab/n - f - \sqrt{3 n \ln n}) \\ 
&& + (ab/n - \sqrt{2 n \ln n})\\
& \geq & a - b - f - 4 \sqrt{n \ln n} \ . 
\end{eqnarray*}
Thus, since $a - b \geq f + 4 \sqrt{n \ln n}$,
we also get that $a' - b' \geq 0$, which implies that the majority value stays the same.
\end{proof}

Next, consider the duplication phase.
As above, let $a$ and $b$ be the number of respective honest nodes with $A$ and $B$ values 
just before the phase starts, 
and $a'$ and $b'$ after~the~phase~ends.

\begin{lemma}
\label{lemma:duplication}
\cbr{For $n \geq e^8$,} In the duplication phase, if $|a - b| \geq f + 4 \sqrt{n \ln n}$,
and the number of empty nodes is \cb{$\e \geq n/4$},
then with probability at least $1 - \cb{2/n}$ 
it holds that the majority value stays the same at the end of duplication phase such that
$$|a'-b'| \geq |a - b|(1 + \e/n) - f - 4 \sqrt{n \ln n} \ .$$
\end{lemma}

\begin{proof}
\cb{Without loss of generality suppose that $a \geq b$.}
Consider a node $u$ which is empty at the beginning of the duplication phase 
(just before the first exchange of the first sub-phase).
When an empty node $u$ enters the second sub-phase of the duplication phase, 
in the first exchange $\{u,v\}$ in the sub-phase $u$ decides whether to copy the value of $v$
which $v$ held just before the start of the phase.

The probability that $v$ has the value $A$ is $a/n$.
Moreover, if $v$ is faulty then $v$ may pretend that it has the value $A$ causing a false duplication 
of value $A$ to $u$.
This occurs with probability $f/n$.
Overall, with probability  $p$,
where \cb{$a/n \leq p \leq (a+f)/n$},
node $u$ will have its empty value~replaced~to~$A$.

Let $\e$ be the number of empty honest nodes just before the phase starts. 
Let $\e_a$ be the number of empty nodes that adopt value $A$ after the duplication phase.
\cb{
Since $a - b \geq f + 4 \sqrt{n \ln n}$, we have $a \geq 4 \sqrt{n \ln n}$.
Consider the expected value $\mu = E[\e_a] = \e p$.
Let $\delta =  \sqrt{2 \ln n} / \sqrt{\mu}$.
Since $\mu = \e p \geq \e a/n \geq n/4 \cdot 4 \sqrt{n \ln n} \cdot 1/n \geq \sqrt{n \ln n}$,
we get $0 < \delta \leq 1$ for $n \geq e^4$.
Hence, from Lemma \ref{lemma:chernoff},
since $\e (a + f) \leq n^2$,
we get
\begin{eqnarray}
\e_a 
& \geq & \mu - \sqrt{2 \mu \ln n} \nonumber
 \geq \e a/n - \sqrt{2 e(a+f) (\ln n)/n} \nonumber\\
& \geq & \e a/n - \sqrt{2 n \ln n} \label{eqn:ea_lower_bound}
\end{eqnarray}
with probability at most $e^{-\delta^2 \mu / 2} \leq e^{-\ln n} = 1/n$.

Let $\e_b$ be the number of empty nodes that adopt value $B$ after the duplication phase.
Suppose that $b \geq \sqrt{3 n \ln n}$.
Consider the expected value $\mu = E[\e_b] = \e p$,
where $b/n \leq p \leq (b+f)/n$.
Let $\delta =  \sqrt{3 \ln n} / \sqrt{\mu}$.
Since $\mu = \e p \geq \e b/n \geq n/4 \cdot \sqrt{3 n \ln n} \cdot 1/n \geq 1/4 \cdot \sqrt{3 n \ln n}$,
we get $0 < \delta \leq 1$ for $n \geq e^8$.
Hence, from Lemma \ref{lemma:chernoff},
since $\e(b+f) \leq n^2$,
we get
\begin{eqnarray}
\e_b 
& \leq & \mu + \sqrt{3 \mu \ln n} \nonumber
 \leq  \e(b+f)/n + \sqrt{3 \e(b+f) (\ln n)/n} \nonumber\\
& \leq &  \e b/n + f + \sqrt{3 n \ln n} \label{eqn:eb_upper_bound}
\end{eqnarray}
with probability at most $e^{-\delta^2 \mu / 3} = e^{-\ln n} = 1/n$.
Note that inequality \ref{eqn:eb_upper_bound} holds even if 
$b < \sqrt{3 n \ln n}$, 
since any of the $b$ nodes can duplicate their $B$ value at most once
(which gives $\e_b < \sqrt{3 n \ln n}$).}

Since $\e_a = a' - a$ and $\e_b = b' - b$,
\cb{from Equations~\ref{eqn:ea_lower_bound} and~\ref{eqn:eb_upper_bound}}
we get
\begin{eqnarray*}
a'-b'
& = & (a + \e_a) - (b + \e_b)
 =  a - b + \e_a - \e_b \\
& \geq & a - b + (\e a/n - \sqrt{2 n \ln n})\\
&& - (\e b/n + f + \sqrt{3 n \ln n}) \\
& \geq & (a - b)(1 + \e/n) - f - 4 \sqrt{n \ln n} \ . 
\end{eqnarray*}
Moreover, since $a - b \geq f + 4 \sqrt{n \ln n}$ and $\e/n \geq 0$,
we also get that $a' - b' \geq 0$, which implies that the majority value stays the same.
\end{proof}

We are now ready to prove Theorem \ref{theorem:algo1}.
\cbr{(The theorem holds for
$n \geq x \ln^2 x$ where $x = 2 \xi^4_1 \xi^2_2$, $\xi_1 = 256$, $\xi_2 = 4$.)}
%

%

%
\begin{proof}[Proof of Theorem \ref{theorem:algo1}]
Let  $\gamma = \xi_1 \xi_2 = 256 \cdot 4$ and $c' = 6(4 \gamma/3 - 1) $.
Without loss of generality, suppose that the majority value is $A$
so that $d = a - b = c' (f + 4\sqrt{n \ln n})$.

The nodes go through the cancellation phase $\gamma$ times.
let $a_i$ and $b_i$ denote the number of nodes with 
respective preferred values $A$ and $B$ at the end of the $i$th cancellation phase,
where $1 \leq i \leq \gamma$.
Suppose that initially $a > n/\xi_2$ and $b > n/\xi_1$.
Hence, $a b/n >  n /(\xi_1 \xi_2)$
\cb{and for sufficiently large $n$, $ab \geq 3n \ln n$}.
From Lemma \ref{lemma:cancellation},
after we apply the cancellation phase once,
we get that $a_1 < a - n /(\xi_1 \xi_2) + \sqrt{2 n \ln n}$,
$b_1 < b - n /(\xi_1 \xi_2) + \sqrt{2 n \ln n}$,
and $a_1 - b_1 \geq a - b - f - 4 \sqrt{n \ln n}$.

As long as $a_j > n/\xi_2$ and $b_j > n/\xi_1$, for $1 \leq j < i$,
after applying the cancellation phase $i$ times
we get that $a_i < a - i(n /(\xi_1 \xi_2) - \sqrt{2 n \ln n})$ and
$b_i < b - i(n /(\xi_1 \xi_2) - \sqrt{2 n \ln n})$.
Moreover, \cb{for $i = \gamma$},
\begin{eqnarray}
\label{eqn:gamma-diff}
a_\gamma - b_\gamma \geq a - b - \gamma(f + 4 \sqrt{n \ln n}) \ .
\end{eqnarray}

Note that since $\gamma < c'$, the precondition of Lemma \ref{lemma:cancellation}
on $|a - b|$ still holds for $i = \gamma$.
Since $\gamma = \xi_1 \xi_2 = 256 \cdot 4$, we get  $\gamma \cdot n /(\xi_1 \xi_2) = n$.
For sufficiently large $n$,
$\gamma  \sqrt{2 n \ln n}  = 256 \cdot 4 \cdot \sqrt{2 n \ln n} \leq n / 256 = n / \xi_1$. 
Hence, \cb{for $i = \gamma$}, $a - \gamma (n /(\xi_1 \xi_2) + \sqrt{2 n \ln n}) \leq a - n + n/\xi_1 \leq n/\xi_1$
and $b - \gamma (n /(\xi_1 \xi_2) + \sqrt{2 n \ln n}) \leq b - n + n / \xi_1 \leq n/\xi_1$.
Therefore, after $\gamma$ iterations it has to be that $b_\gamma \leq n/\xi_1$  or  $a_\gamma \leq n/\xi_2$ (since $\xi_2 \leq \xi_1$).

The nodes continue with the resolution phase.
By
Lemma~\ref{lemma:resolution},
if $b_{\gamma} \leq n / \xi_1 $ and $a_\gamma \geq n/\xi_2$
then all honest \cbr{(undecided)} nodes will decide the majority value $A$ in this phase.
\cbr{Otherwise, the undecided nodes continue to the duplication phase
(with the help of any decided nodes).
Note that from Lemma~\ref{lemma:resolution},
any nodes that decide choose the correct majority value in their decision.}

Therefore, suppose that $b_{\gamma} \leq a_\gamma \leq n/\xi_2 = n/4$
and that the nodes continue to the duplication phase.
Then, the number of non-empty honest nodes are at most $a_\gamma + b_\gamma \leq n/2$,
and the number of empty nodes is \cb{$\e \geq n - n/2 - f = n/2 - f \geq n/4$, since $f \leq n/c_f = n /256$}.
Let $a'$ and $b'$ be the number of honest nodes with respective values $A$ and $B$. 
\cb{From Lemma~\ref{lemma:duplication}
and Equation~\ref{eqn:gamma-diff}, 
since $c'/6 = (4 \gamma/3 - 1)$:}
\begin{eqnarray*}
a'-b' 
& \geq & (a_\gamma - b_\gamma)(1 + \e/n) - f - 4 \sqrt{n \ln n} \\
& \geq & (a_\gamma - b_\gamma)(3/2 - f/n) - f - 4 \sqrt{n \ln n}\\
& \geq & (a_\gamma - b_\gamma)(3/2 - 1/256) - f - 4 \sqrt{n \ln n}\\
& \geq & (a_\gamma - b_\gamma)4/3 - f - 4 \sqrt{n \ln n}\\
& \geq & (a - b - \gamma(f + 4 \sqrt{n \ln n})) 4/3 - f - 4 \sqrt{n \ln n}\\
& \geq & \frac{4}{3} (a-b) - (4\gamma/3 - 1) (f + 4 \sqrt{n \ln n})\\
& \geq & \frac{4}{3} (a-b) - \frac{1}{6} (a-b)\\
& = & \frac {7} {6} (a-b) \ .
\end{eqnarray*}

\cbr{
Therefore, after $i$ repetitions 
the difference 
is at least $(7/6)^i(a-b)$.
For $i = \lceil \log_{7/6} (n/4) \rceil = O(\log n)$,
the difference is at least $n/4$.
Thus, at least $n/4 = n / \xi_2$ nodes have value $a$,
which implies 
that at the end of the last cancellation phase,
at most $n/\xi_1$ nodes have value $b$.
Therefore, from Lemma \ref{lemma:resolution},
the honest nodes will reach the correct decision on the majority.
Since each repetition requires $\gamma + 2 = O(1)$ phases,
the algorithm requires $O(\log n)$ phases.
From Lemma \ref{lemma:drift},
since each phase requires $O(\log^2 n)$ parallel time,
we have in total $O(\log^3 n)$ parallel time.
}

Combining the probabilities of all the involved lemmas for each of the $O(\log n)$ phases, where each phase fails with probability at most $O(1)/n$, 
the result holds with probability at least $1 - O(\log n) /n$. 
\cbr{The algorithm requires in total $O(\log^3 n)$ states per node since 
all 
local variables have maximum absolute value of $O(\log^3 n)$.}
\end{proof}

%% file: algoprecise.tex
\section{Algorithm \scfd}
\label{sec:scfd}

\begin{algorithm}[t]
{\footnotesize \cbr{
\caption{\sf Symmetric-C-Full-D}
\label{alg:scfd}




\tcp{Cancellation Phase Actions}
\If{$($either $u$ or $v$ are in second subphase$)$ and 
$((value_u = A \wedge value_v = B) \vee (value_u = B \wedge value_v = A))$} 
{$value_u \gets \bot$; \tcp{$u$ cancels its value, and $v$ also cancels its own value}}
\BlankLine
\tcp{Duplication Phase Actions}
\If{$($$u$ is in second subphase$)$ and 
$(value_u \neq \bot)$ and $(value_v = \bot)$ and $($$u$ was not empty at the beginning of the duplication phase$)$ and $($$u$ has not cloned its value yet to any node in this phase$)$\linebreak} 
{$value_v \gets value_u$; \tcp{$v$ adopts $u$'s value}}
}
}
\end{algorithm}
\cbr{
Algorithm \scfd\ (Algorithm \ref{alg:scfd})
follows the structure of skeleton Algorithm~\ref{alg:skeleton}.
We fill-in the omitted action steps for the cancellation and duplication phases.
(This algorithm is not using the $saved_u$ variables.)

\noindent
{\bf\em Cancellation phase:} 
The cancellation phase is executed just once ($\gamma = 1$)
before continuing to the resolution phase.
In the cancellation phase,
}
a non-empty node will attempt for the whole duration of the phase to cancel its value. Suppose that $u$ participates in an exchange $e_i = \{u,v\}$ where $u$ is in the second subphase while $v$ is also is in the cancellation phase (in any subphase). Then, $u$ compares $value_u$ with $value_v$. If $value_u$ is non empty and opposite than $value_v$ (the opposite of $A$ is $B$, and vice-versa) then $u$ cancels its value and $u$ becomes an empty node by setting $value_u \gets \bot$. The same also happens symmetrically to node $v$.  If node $u$ was not successful in cancelling its value (i.e. $value_v$ is empty or not opposite) then node $u$ will try again in the next exchange in the same phase and will keep trying until the end of the second subphase. 

\noindent
{\bf\em Duplication phase actions:} Suppose that $u$ participates in an exchange \cbr{$e_i = \{u,v\}$} in the second subphase such that $value_u \neq \bot$ ($u$ is non-empty) and $value_v = \bot$. If $u$ was non-empty at the beginning of the phase and has not cloned its value before in this phase, then $u$ will attempt to duplicate (copy) its value to $v$. If the copying is unsuccessful (because $value_v$ is non-empty) then node $u$ will try to duplicate again in the next exchange of the current phase and will keep trying until the end of the second subphase.


\subsection{Analysis of Algorithm \scfd}

We start with some basic results that will help us to argue about the preference tally difference in the phases.  
The first result bounds the number of honest nodes that can be affected by the faulty nodes during a parallel time step.

\begin{lemma}
\label{lemma:honest-Byzantine}
For $f \leq n/4$,
out of $n$ exchanges there are at most $3(f + \ln n)$ exchanges between faulty and honest nodes,
with probability at least $1 - e^{-f/8}/n$.
\end{lemma}

\begin{proof}
If $f = 0$,
\cb{then there are no exchanges between honest and faulty nodes,
and the result holds trivially.}
So assume that $1 \leq f \leq n/4$.

The number of pairs between faulty and honest nodes is $f (n-f)$.
\cb{Since the number of exchange pairs is ${n \choose 2} = n(n-1)/2$}, 
the probability that an exchange is between a faulty and an honest node is \cb{$p = f (n-f) /(n(n-1)/2) = 2 f (n-f) / (n^2 - n)$}.

Consider now a sequence of $n$ exchanges.
Let $X_i$ denote the random variable such that $X_i = 1$,
$1 \leq i \leq n$, 
if the $i$th exchange in the sequence is between a faulty and an honest node, 
and $X_i = 0$ otherwise.
Let $X = \sum_{i=1}^n X_i$ denote the total number 
of exchanges between faulty and honest nodes in the sequence.
We have $\mu = E[X] = n p  = \cb{2f(n-f)/(n-1)}$.
We can bound $\mu$ from above as 
$$\mu = \cb{2f(n-f)/(n-1)} \leq 2f (n-1) /(n-1) = 2f \ ,$$
and from below as 
$$\mu = \cb{2f (n-f)/(n-1)} \geq 2 f (n - n/4) / n = 3 f / 2 \ .$$
Hence,
$$\frac {3f} {2} \leq \mu \leq 2f \ .$$

\cb{Let
$\delta = 1/2 + k/\mu$, for some $k \geq 0$. 
We have
$\Pr[X \geq 3 \mu /2 + k] = \Pr[X \geq (1 + \delta) \mu]$.
If $k/\mu \leq 1/2$, then $\delta \leq 1$,
and from Lemma \ref{lemma:chernoff},
\begin{eqnarray*}
\Pr[X \geq (1 + \delta) \mu]
& \leq & e^{-\frac{\delta^2 \mu} 3}
 =  e^{-\frac{(1/2 + k / \mu)^2 \mu} 3}\\ 
 & =  & e^{-\frac{(1/4 + k^2 / \mu^2 + k/\mu) \mu} 3}
 \leq   e^{-\frac{(1/4 + k/\mu) \mu} 3}\\
& = & e^{-(\mu/12 + k/3)}
 \leq  e^{-f/8 - k/3}
\ .
\end{eqnarray*}
On the other hand,
if $k/\mu > 1/2$, then $\delta > 1$,
and from Lemma \ref{lemma:chernoff},
\begin{eqnarray*}
\Pr[X \geq (1 + \delta) \mu]
& \leq & e^{-\frac{\delta \mu} 2}
 =  e^{-\frac{(1/2 + k / \mu) \mu} 2}\\ 
& = & e^{-(\mu/4 + k/2)}
 \leq  e^{-3f/8 - k/2}\\
& \leq & e^{-f/8 - k/3}
\ .
\end{eqnarray*}
Thus,
$\Pr[X \geq (1 + \delta) \mu] \leq e^{-f/8 - k/3}$.
Therefore, 
since $\mu \leq 2 f$,
we get $\Pr[X \geq 3 f + k] \leq \Pr[X \geq 3 \mu /2 + k] \leq e^{-f/8 - k/3}$.
Setting $k = 3 \ln n$,
we get that $\Pr[X \geq 3 f + k] \leq e^{-f/8 - \ln n} = e^{-f/8}/n$, as needed.}
\end{proof}

Next, we bound the number of exchanges involved in any phase. 
For this result we consider the number of exchanges during which any node 
can be in the specific~phase. 

\begin{lemma}
\label{lemma:phase-exchanges}
Any specific phase $\phi$ involves at most $z = nD/2 + n\zeta + 1$ consecutive exchanges of the system with probability at least $1 - 4/n^2$.
\end{lemma}

\begin{proof}
Consider a phase $\phi \geq 0$.
Suppose that $u$ is the first node that enters the phase, such that exchange $e_i$
involves $u$ and causes the updated local counter $C_u$ to indicate the beginning of the phase.
Assume for simplicity that there is an exchange $e_0 = \{u,u\}$ for the case $\phi = 0$.
Hence, $\lfloor C_u / D \rfloor = \phi$.
From the proof of Lemma \ref{lemma:drift}, 
using $r = i$ and $\mu = 2 r / n = 2 i / n$,
we get $\Pr[|2i/n - C_u| \leq \zeta] \geq 1 - 2/n^2$.
Hence, with probability at least $1 - 2/n^2$ for exchange $e_i$ it holds that 
$2i/n \geq C_u - \zeta$. 

Let $v$ be the last node to exit phase $\phi$, such that exchange $e_j$
is the last involving $v$ in phase $\phi$
which causes an update to $C_v$.
Hence, right after exchange $e_j$, $\lfloor C_v / D \rfloor = \phi$.
Similarly as above, 
from the proof of Lemma \ref{lemma:drift}
with probability at least $1 - 2/n^2$ for exchange $e_j$ it holds that 
$2j /n \leq C_v + \zeta$.

Since $\lfloor C_u / D \rfloor = \lfloor C_v / D \rfloor = \phi$ ,
we get that $C_v - C_u \leq D$. 
Therefore,
with probability at least $1 - 4/n^2$, the maximum number 
of consecutive exchanges involved in phase $\phi$ is 
at most
\begin{eqnarray*}
j - i + 1 
& \leq & n(C_v + \zeta)/2 - n(C_u - \zeta)/2 + 1 \\
& \leq & n(C_v - C_u)/2 + n\zeta + 1\\
& \leq & nD/2 + n\zeta + 1 \ .
\end{eqnarray*}
\end{proof}

We now consider how many system exchanges occur 
while the local counter of a node is incremented at least~$12 \ln n$~times.

\begin{lemma}
\label{lemma:increment-exchanges}
While the local counter of a node is incremented $y \geq 12 \ln n$ times,
more than $ny/4$ system exchanges occur with probability at least $1 - 1/n^2$.
\end{lemma}

\begin{proof}
Consider a sequence of exchanges $e_{i_1}, e_{i_2}, \ldots, e_{i_k}$ in which the local counter of a node $u$ 
is incremented at least $y$ times.
Suppose for the sake of contradiction that $k \leq n y / 4$.
If necessary, extend the sequence of $k$ exchanges to contain exactly $k' = ny/4$ exchanges.  
An exchange involves node $u$ with probability $(n-1)/{{n} \choose {2}} = 2/n$.
The expected number of exchanges involving $u$ is $\mu = k' \cdot 2/n = ny/4 \cdot 2/n = y/2$.
Let $X$ be the actual number of exchanges involving node $u$ during the $k'$ exchanges.
Since $2\mu = y$, using Lemma \ref{lemma:chernoff} for $\delta = 1$,
$Pr[X \geq y] = Pr[X \geq 2 \mu ] = Pr[X \geq (1 + \delta)\mu] \leq e^{-\mu/3} = e^{-y/6} \leq  e^{-2 \ln n} = n^{-2}.$
Thus, with probability at least $1-1/n^2$ the local counter of $u$ was increment less than $y$ times.
However, since during the $k$ exchanges the local counter of $u$ was incremented at least $y$ times,
it must be that  $k > ny/4$ with probability at least $1 - 1/n^2$, as needed.
\end{proof}

In the lemma below we assume that the respective tally for values $A$ and $B$ 
before the phase is $a$ and $b$,
while the tally after the phase is $a'$ and $b'$.
Without loss of generality assume that the majority is value $A$.
The parameter $z$ is as specified in the statement of Lemma \ref{lemma:phase-exchanges}.

\begin{lemma}
\label{lemma:cancellation2}
In the cancellation phase, if $a - b \geq 3z(f + \ln n)/n$, and $D \geq 8 \cdot 12 \ln n$,
then with probability at least $1 - O(\log^2 n)/n$ it holds that
either 
\begin{itemize}
\item
$a' \geq n/8$ and $b' = 0$, or 
\item
$a' < n/8$ and $a' - b' \geq a - b - 3z(f + \ln n)/n \geq 0.$
\end{itemize}
\end{lemma}

\begin{proof}
The second subphase of a node involves $D/3$ increments of its local counter.
From Lemma \ref{lemma:increment-exchanges},
the second subphase involves more than $n D/12 \geq 8 n \ln n$ exchanges,
with probability $q \geq 1 - 1/n^2$.

First, suppose that $a' \geq n/8$.
A node $u$ with value $B$ will cancel its value if it interacts with one of the nodes 
with value $A$, which are at least $n/8$ throughout the phase, or if it interacts with a faulty node.
The probability that an exchange $e$ of the cancellation phase is such that $e = \{u,v\}$ 
where $v$ has value $A$, resulting to the cancellation of the $B$ value of $u$,
is at least $(n/8) / {n \choose 2} = 1 / (4 (n-1)) > 1 / (4 n)$.
The probability that node $u$ does not participate in any cancellation exchange 
during the at least $8 n \ln n$ exchanges (given with probability $q \geq 1 - 1/n^2$) 
of its second subphase
is at most $(1 -  1/(4n))^{8n\ln n} \leq e^{-2 \ln n} \leq n^{-2}$.
Thus, the probability that node $u$ is unsuccessful in the cancellation is at most $(1-q) + n^{-2} = 2 n^{-2}$.
Hence, the probability that any of the $b$ nodes with value $B$ is unsuccessful 
is at most $2 b n^{-2} \leq 2 n n^{-2} = 2 n^{-1}$.
Therefore, with probability at least $1 - 2 / n$,
all the $b$ nodes have their values successfully cancelled,
and $b' = 0$.

Now, suppose that $a' < n/8$.
From Lemma \ref{lemma:honest-Byzantine},
the probability that in a parallel time step (involving $n$ exchanges)
there are more than $3(f + \ln n)$ exchanges between honest and faulty nodes is at most $e^{-f/8}/n$.
From Lemma \ref{lemma:phase-exchanges}, a phase involves $z = O(n \log^2 n)$ exchanges
(considering the participation of any of the $n$ nodes in the phase).
Thus, in the $z/n = O(\log n)$ parallel time steps the probability that in any of the $z/n$ time steps
there are more than $3(f + \ln n)$ 
exchanges between honest and faulty nodes is at most 
$z/n \cdot e^{-f/8}/n = O(\log^2 n)/ n$.
Hence,
with probability at least $1 - O(\log^2 n) / n$,
during the $z/n$ parallel time steps the number of exchanges between honest and faulty nodes
are at most $3z(f + \ln n)/n$.
This implies that the number of cancellations between the faulty nodes and the 
nodes that have value $A$ are at most $3z(f + \ln n)/n$.
Such cancellations result to decreasing the difference between
the number of nodes with different values.
Hence, $a' - b' \geq a - b - 3z(f + \ln n)/n \geq 0$,
with probability at least $1 - O(\log^2 n) / n$.

\cbr{The result follows by combining the two cases described above for $a' \geq n/8$ and $a < n/8$
and their 
probabilities.}
\end{proof}

Now we consider the duplication phase.
Let $a$ and $b$ are the number of respective honest nodes with $A$ and $B$ values 
just before the phase starts, 
and $a'$ and $b'$ after the phase ends.
Again, without loss of generality assume that $A$ is the majority value.

\begin{lemma}
\label{lemma:duplication2}
In the duplication phase, if $3z(f + \ln n)/n + f \leq n/4$ and
$a \leq n/8$, then with probability at least $1 - O(\log^2 n)/n$ it holds that
$a' \geq n/8$ and $a' - b' \geq 2(a - b) - 3z(f + \ln n)/n$.
\end{lemma}

\begin{proof}
Similar to the analysis in Lemma \ref{lemma:cancellation2},
the number of faulty nodes that interact with honest ones during the cancellation phase
is at most $3z(f + \ln n)/n$ with probability at least $1 - O(\log^2 n) / n$.
During these interactions empty honest nodes may adopt a value.
Therefore, up to $3z(f + \ln n)/n + f \leq n/4$ (where the additional $f$ are the original faulty nodes), that were not counted in $a$ or $b$ become non empty due to faulty nodes.
Thus, the number of empty honest nodes which are not affected from the faulty nodes are
at least $n - 2 \cdot n/8 - n/4 = n/2$.
If all honest nodes that are counted at the beginning of the phase in $a$ and $b$ are duplicated in this phase,
then that leaves at least $n/2 - n/4 = n/4$ empty nodes.

For any non-empty node $u$ counted in $a$ or $b$,
the probability that an exchange $e$ of the duplication phase is such that $e = \{u,v\}$ 
where $v$ is empty, resulting to the duplication of the value of $v$,
is at least $(n/4) / {n \choose 2}  = 1/(2(n-1)) > 1 / (2 n)$.
As stated in the proof of Lemma \ref{lemma:cancellation2},
the second subphase involves at least $8n \ln n$ exchanges, with probability $q \geq 1 - 1/n^2$.
The probability that node $u$ does not participate in any duplication exchange 
during these $8 n \ln n$ exchanges of the second subphase 
is at most $(1 - 1 /(2n))^{8n\ln n} \leq e^{-4 \ln n} \leq n^{-4}$.
Thus, the probability that node $u$ is unsuccessful in the duplication is at most  $(1-q) + n^{-4} \leq 2 n^{-2}$.
Hence, the probability that any of the $a$ or $b$ nodes is unsuccessful to duplicate
is at most $(a+b) n^{-2} \leq 2 n n^{-2} = n^{-1}$.
Therefore, with probability at least $1 - n^{-1}$,
all the $a$ and $b$ nodes have their values successfully duplicated.

If we do not consider the impact of the faulty nodes,
then  $a' - b' \geq 2(a - b)$.
However, the faulty nodes may have converted up to $3z(f + \ln n)/n$ honest empty nodes to non-empty (with probability at least $1 - O(\log^2 n) / n$),
which decreases the difference between the honest nodes with values to
$a' - b' \geq 2(a - b) - 3z(f + \ln n)/n$.
\end{proof}


We are now ready to give the proof of Theorem~\ref{theorem:algo2}.
\cbr{(The Theorem holds for $n \geq e^2$, a consequence of Lemma \ref{lemma:drift}.)} 

%
\begin{proof}[Proof of Theorem \ref{theorem:algo2}]
First consider the case that $d = |a - b| = \Omega(1 + f \log^2 n)$.
Let $\delta = 3z(f + \ln n)/n$.
Without loss of generality assume that $A$ is the majority value and 
$a - b \geq 6 \delta$.
First the nodes execute the cancellation phase. Let $a'$ and $b'$ be the resulting values at the end of the cancellation~phase.

\cbr{Next is the resolution phase.
We use Lemma~\ref{lemma:resolution},}
but with the related constants adjusted so that $\xi_2 = 8$ and $\xi_1 = 512$.
The other constants are: $c_f = 512$, $\xi = 64, c_{\sigma_1} = 12, c_{\sigma_2} = 8 c_{\sigma_1}, c_\psi = 128 c_{\sigma_1}$.
This allows the honest nodes to decide on the majority value with the appropriate sampling if the tally for the majority is at least $n/8$,
while the tally for the minority is at most $n/512$.

\cbr{From the cancellation phase \cbr{(Lemma \ref{lemma:cancellation2})}, there are two possibilities when reaching the resolution phase:} 
\begin{itemize}
\item
$a' \geq n/8$ and $b'=0$:
by
Lemma \ref{lemma:resolution} all honest nodes decide 
the majority value $A$ with probability at least~$1 - 4/n$.
\item
$a' < n/8$: in this case, $a' - b' \geq a - b - \delta \geq 0$,
which implies that $A$ remains the majority. 
\cbr{In this case, some of the nodes may not decide.
Nevertheless, by Lemma~\ref{lemma:resolution},
the nodes that are able to decide make the correct choice for the majority value.}
\end{itemize}

\cbr{The undecided nodes (with the help of decided nodes) 
continue to the duplication phase.}
Let $a''$ and $b''$ be the tallies for the two values $A$ and $B$
at the end of the duplication phase.
\cbr{Since $a' < n/8$,
from Lemma \ref{lemma:duplication2} we get}
\begin{eqnarray*}
a'' - b'' & \geq & 2(a' - b') - \delta
\ \ \geq \ \ 
2(a-b - \delta) - \delta
\\ 
&=& 
2(a-b) -3\delta
\ \ \geq \ \ 
2 (a - b) - \frac{a - b} {2}\\
& = & \frac {3}{2} (a - b) \ .
\end{eqnarray*}
\cbr{Therefore, after $i$ repetitions of the above process, 
the difference 
becomes $(3/2)^i(a-b)$.
From Lemma \ref{lemma:resolution},
for $i = \lceil \log_{3/2}(n/8) \rceil = O(\log n)$,
the honest nodes make the correct decision on the majority,
since $(3/2)^i(a-b)\geq n/8 = n/\xi_2$.
Since each repetition requires three phases (recall $\gamma = 1$),
the whole algorithm requires $O(\log n)$ phases.
From Lemma \ref{lemma:drift},
each phase requires $O(\log^2 n)$ parallel time,
which gives in total $O(\log^3 n)$ parallel time.
}

Combining the probabilities of all the involved lemmas for each of the $O(\log n)$ phases, 
the result holds with probability at least $1 - O(\log^3 n) /n$.
The algorithm requires in total $O(\log^3 n)$ states per node since all the constant number of local variables have a maximum absolute value of $O(\log^3 n)$. 

Now consider the case $d = |a-b| = \Omega(n)$.
The analysis above covers all cases except  $6\delta > n$.
This case can be covered by considering $a-b = n$.
We use a small variant of the algorithm where the first phase to execute is the resolution phase.
After that, the phases execute as normal (i.e. cancellation, resolution, duplication, and then repeat).
Since $a-b = n$, 
the number of honest nodes that have the majority value are at least $n - n/c_f > n/\xi_2$,
while the number of nodes that have the minority value are at most the faulty nodes $n/c_f \leq n /\xi_1$.  
From Lemma \ref{lemma:resolution}, a correct majority decision is reached in the first phase.
Hence, the result follows considering all cases for $d$.
\end{proof}

%% file: combined.tex
\section{Combining the two algorithms}
\label{sec:combined}

Suppose that $f$ is not known, therefore nodes cannot a priori decide which algorithm to choose.
In order to make use of the best of the two algorithms,
one may want to run them together and choose the one that guarantees correct majority.
This however cannot be done in a simple way --
if the two algorithms return different values at a node,
it may not be able to decide which one is correct (recall that it does not know $f$).
Therefore, we combine the two algorithms, \acpd\ \cb{(Algorithm \ref{alg:acpd})} and \scfd\ \cb{(Algorithm \ref{alg:scfd})},
into one \combined, as~follows.

In the beginning, \combined\ checks if the bias between the two values could be (likely) larger than $cn/\log^2 n$, for some suitable constant $c>0$. More precisely, in the first $\log^3 n$ interactions, it sends its initial value and receives a value from the other peer (which might be Byzantine). It counts the number of values $A$ and $B$, and if the difference is at least $c\log n$, it sets variable $Z_0$ to $1$, otherwise to $0$.

Next, the algorithm lets the two algorithms, \acpd\ and \scfd, run in parallel three times in every node.
Let $c>0$ be a sufficient constant, which value will arrive from the analysis (proof of Theorem~\ref{theorem:combined}).
The first run assumes initial values at every node.
Before starting the second run, 
the algorithm creates as expected bias of $(1/c)\sqrt{n\log n}$ initial values for $A$;
more precisely, it converts $(1/c)\sqrt{n\log n}$ nodes, on expectation, with initial value $B$ to have value $A$. This is done by tossing a random coin
by each node with initial value $B$,
with probability of conversion $c\sqrt{\frac{\log n}{n}}$.
Before starting the third run, 
the algorithm creates an expected bias of $(1/c)\sqrt{n\log n}$ initial values for $B$;
more precisely, it converts $(1/c)\sqrt{n\log n}$ nodes, on expectation, with initial value $A$ to have value~$B$.
The conversion is done similarly as before the second run:
by tossing a random coin by each node with initial value $A$,
with probability of conversion~$c\sqrt{\frac{\log n}{n}}$.

After finishing all three runs of both algorithms, a node checks its own decisions for every run of the two algorithms.
Let $(X_i,Y_i)$ be the decision made at the end of run $i$, for $i=1,2,3$, where $X_i$ is the decision of algorithm \acpd\ 
and $Y_i$ is the decision of algorithm \scfd.
Then the node does the following:
\begin{itemize}
\item if $Z_0=1$ then answer $X_1$;
    \item 
if $Z_0=0$ and $X_2 = X_3$ then the final decision is $X_1$ (we show later that in fact $X_1 = X_2 = X_3$ with high probability);
\item
if $Z_0=0$ and $X_2 \neq X_3$, then the final decision is $Y_1$.
\end{itemize}

\noindent
{\bf Random common coin implementation:}
One could implement the random coin needed for algorithm
\combined\ using local random bits: each node 
keeps tossing locally its symmetric coin, independently, until its counter
reaches $\log\sqrt{(1/c)^2 n\log n}$ or $0$ is drawn.
In the former case, it outputs $1$ (corresponding to ``conversion''), in the latter -- it outputs $0$ (corresponding to ``no conversion'').
Note that the output is $1$ with probability $(1/2)^{\log\sqrt{(1/c)^2 n\log n}}=c\sqrt{\frac{\log n}{n}}$ and $0$ otherwise.
Also note that it is implementable in the model, as
it only requires storing a counter ranging from $1$ to $\log\sqrt{(1/c)^2 n\log n}$, which could be done in $O(\log\log n)$ local memory.


%
\begin{proof}[Proof of Theorem \ref{theorem:combined}]
Let $n$ be sufficiently large and $c>0$ be a sufficient 
constant to satisfy asymptotic notations. 
First observe that if $d = |a-b|> 2c n/\log^2 n$ 
then $Z_0=1$, 
and if $|a-b|< (c/2) n/\log^2 n$ 
then $Z_0=0$, both  with high probability, by Chernoff bounds, c.f., Lemma~\ref{lemma:chernoff}.
%

Due to Lemma~\ref{lemma:drift} and the fact that we run only
three executions in a row, we keep the required synchronization of phases with high probability. Therefore we can apply the analysis' of algorithms \acpd\ and \scfd\ done in Sections~\ref{sec:acpd} and~\ref{sec:scfd}, respectively.

The performance, number of states and probability follow from Theorems~\ref{theorem:algo1} and~\ref{theorem:algo2}, respectively for algorithms \acpd\ and \scfd, and from the implementation of random common coin;
we apply union bound~to~
all those
probabilities.

It remains to prove correctness if the lower bound formula for $|a - b|$ holds.
W.l.o.g. we may assume that $a>b$, as the algorithm is oblivious with respect to importance of any of the values over the other.
We know, from assumptions in the Theorem, that $|a - b|= \Omega(\min\{f + \sqrt{n \log n}, f \log^2 n + 1, n\})$, which implies that $|a - b|> (c/2) \sqrt{n \log n}$ if constant $c$ is chosen appropriately.

Assume first that $|a - b|>5c\sqrt{n \log n}$ holds.
We have two sub-cases. 
If $|a-b|> 2c n/\log^2 n$ then $Z_0=1$ and $X_1$ will be returned, which is correct with sufficient probability due to Theorem~\ref{theorem:algo1}.
Otherwise, for $|a-b|> 2c n/\log^2 n$, we may have $Z_0=1$ or $Z_0=0$, but as we show the other properties will assure correct answer.
In the second run this difference cannot decrease, and thus the result of algorithm \acpd\ will be the same in the first two runs, i.e., $X_1=X_2$ with the desired probability.
Before the third run, at most $2n\cdot c\sqrt{\frac{\log n}{n}}=2c\sqrt{n\log n}$ nodes with initial value $A$ may convert to starting value $B$, with high probability (c.f., Chernoff bound in Lemma~\ref{lemma:chernoff}), which decreases the difference $a-b$ by at most $4c\sqrt{n\log n}$, which is still bigger than $5c\sqrt{n \log n}-4c\sqrt{n\log n}= c\sqrt{n\log n}$.
Therefore, algorithm \acpd\ will produce the same result as in the first two runs, i.e., $X_1=X_2=X_3$, with the desired probability. In this case, algorithm \combined\ returns value $X_1$, regardless of $Z_0=1$ or $Z_0=0$, which is correct as the required lower bound on $|a-b|$
in Theorem~\ref{theorem:algo1}, corresponding to algorithm \acpd, holds (hence, the first run of \acpd\ was correct with desired probability).

Next assume that $|a - b|< c\sqrt{n \log n}$.
From the beginning analysis we have that $Z_0=0$ with high probability.
Before the second run, the difference $|a - b|$ cannot decrease, and thus the result of algorithm \acpd\ will be the same in the first two runs, i.e., $X_1=X_2$ with the desired probability, c.f., Theorem~\ref{theorem:algo1}.
Before the third run, at least $(1/2) n\cdot c\sqrt{\frac{\log n}{n}}=(c/2)\sqrt{n\log n}$ nodes with initial value $A$ may convert to starting value $B$, with high probability (c.f., Chernoff bound in Lemma~\ref{lemma:chernoff}), which decreases the difference $a-b$ by at least $c\sqrt{n\log n}$, which means that now we have $b-a > (c/2)\sqrt{n \log n}$.
Therefore, algorithm \acpd\ will produce correct but different majority value in the third run comparing with the previous two runs i.e., $X_1=X_2\ne X_3$, with the desired probability. 
In this case, algorithm \combined\ returns the correct value $Y_1$, as the required lower bound on $|a-b|$ in the corresponding Theorem~\ref{theorem:algo2} holds (hence, the first run of \scfd\ was correct with desired probability).

It remains to consider case $c\sqrt{n \log n} \le a - b\le 5c\sqrt{n \log n}$.
It is easy to see that the assumptions of both Theorems~\ref{theorem:algo1} and~\ref{theorem:algo2} hold, hence no matter if \combined\ outputs $X_1$ or $Y_1$, it will output the correct majority $A$ with sufficient probability.
\end{proof}

It is worth noting that the considered implementation of algorithm \combined\ uses only $O(\log\log n)$ of random bits per node, and works against Fully (dynamic) Byzantine adversary. First, because the adversary cannot prevent the right bias, with high probability, as it is done by local random choices. Second, knowing the bias does not help the adversary either, as the runs are deterministic procedures, and thus the Theorems~\ref{theorem:algo1} and~\ref{theorem:algo2} apply -- 
both
against
Fully (dynamic)~Byzantine~adversary.

\remove{
\paragraph{Creating a bias without local randomness.}
Following the idea of a Synthetic Coin~\cite{Alistarh},
we may try to extract some randomness and create a bias necessary for algorithm \combined\ without using local randomization. In our model, however, the adversary may create bias in such coin. The probabilistic space of honest nodes is not affected, as the coin is produced at each node by flipping the coin after each interaction -- the adversary cannot intervene to this process. However using the coin is affected, as nodes may connect to a Byzantine node that creates bias.
We overcome this problem in the following way. We produce 
}

%% file: related.tex
\section{Related Work}
\label{sec:related}

Most of the related works for population protocols is for failure-free scenarios. We give an overview of these works. Several population protocols in the literature consider the {\em exact-majority voting} version of the 
majority consensus problem \cite{DBLP:conf/podc/AlistarhGV15,AngluinADFP06,AngluinAE2008fast,berenbrink2018population,10.1145/3087801.3087858}, which is the case where the tally difference is $d=1$.
Our algorithms also solve the exact-majority voting problem when $f=0$, that is, when all nodes are honest.
Trade-offs between the required number of parallel time steps 
and the number of states have been studied in \cite{AlistarhAEGR17,DBLP:conf/soda/AlistarhAG18}.
These works explain the lower bound of $\Omega(\log n)$ on the stabilization time for solving the majority consensus \cite{AlistarhAEGR17}, and also the $\Omega(\log n)$ lower bound requirement on the number of states
for algorithms where the discrepancy between the preferred values tally can be small
and the expected time complexity is $O(n^{1-\epsilon})$ \cite{DBLP:conf/soda/AlistarhAG18}.

There are population protocols for the majority consensus problem with a constant number of states \cite{5461999,DBLP:journals/dc/MertziosNRS17}.
However, these protocols require at least linear time to be solved due to the aforementioned lower bound related to the number of states.
Declaring a special node as a leader can significantly improve the convergence time,
as demonstrated in \cite{AngluinAE2008fast} which presents a majority consensus protocol with constant number of states
and $O(\log^2 n)$ time.
By relaxing the notion of time needed to reach a correct state,
recent works study trade-offs that improve on the time complexity while keeping the number of states small (e.g. sub-logarithmic) with appropriate parameterization \cite{Time-space-trade-offs,kosowski2018population}. Population protocols have also been used to solve other coordination problems such as leader election \cite{10.5555/3174304.3175286,10.1145/3323165.3323178}.

Guerraoui and Ruppert~\cite{DBLP:conf/icalp/GuerraouiR09}
consider a variant of the population protocol model, called community protocol, where nodes have IDs and nodes can store a limited number of other nodes’ IDs. In our paper we do not assume that nodes have IDs, nor enough storage to maintain e.g., random IDs (we only require $O(\log \log n)$ bits per node, for $n$ nodes). The work in~\cite{DBLP:conf/icalp/GuerraouiR09} considers the general class of decision problems in class NSPACE$(n \log n)$. In our paper, we focus on the majority consensus problem, which is one of the most popular problems studied in the population protocol literature.
We consider a random scheduler which is the most common assumption in population protocols’ literature,
while the work in~\cite{DBLP:conf/icalp/GuerraouiR09}
considers a general setting of exchanges with a fairness condition that guarantees that nodes are not disconnected.
Moreover, \cite{DBLP:conf/icalp/GuerraouiR09} presents results for $f=O(1)$ number of failures, 
while we present results for $f=O(n)$ failures.
Summarizing, the results in \cite{DBLP:conf/icalp/GuerraouiR09} are very interesting and challenging, but our work approaches the majority consensus problem from a different angle, obtaining significantly bigger tolerance to the number of node failures without using IDs, but for a random scheduler.

The Byzantine Agreement problem, in which honest nodes need to make a decision on one of the input values (but not necessarily a majority), was introduced by Pease, Shostak and Lamport~\cite{DBLP:journals/jacm/PeaseSL80}.
They showed~\cite{DBLP:journals/toplas/LamportSP82,DBLP:journals/jacm/PeaseSL80} that even in  a synchronous message-passing system with Byzantine failures the number of failures needs to be less than $n/3$ for a deterministic solution to exist. 
The $\Omega(fn)$ lower bound on the message complexity of deterministic solutions was proved by Dolev and Reischuk~\cite{DBLP:journals/jacm/DolevR85}
and Hadzilacos and Halpern~\cite{DBLP:journals/mst/HadzilacosH93}, which corresponds to the number of interactions in the population model. 
Similar impact is caused by asynchrony itself, even if there are no Byzantine failures -- linear time or quadratic communication is necessary even for a simpler gossip problem, as shown by Georgiou et al.~\cite{GeorgiouGGK13}.
This quadratic performance occurs because in deterministic message-passing solutions the communication is also deterministic, thus more controlled by the adversary (comparing to the population model, in which communication is scheduled by an independent random scheduler). The best qubic upper bound was given by Kowalski and Moustefoui~\cite{KowalskiM13}.

Another difference between classic message-passing systems and population model, well visible in the context of computing majority -- the protocols in the former model often require huge memory and number of communication bits, c.f., a popular deterministic protocol in~\cite{DBLP:conf/podc/Bar-NoyDDS87} or a randomized one in~\cite{DBLP:conf/podc/KingS10}
The impact of randomization was also considered in classic models, although most of the solutions assumed additional restrictions, such as achieving only an 
almost-everywhere agreement~\cite{DBLP:conf/focs/KingSSV06} or restrictions on adversarial power such as
oblivious adversary c.f.,~\cite{DBLP:conf/podc/Aspnes12,DBLP:conf/stoc/CanettiR93}.

Computing majority in population model is not directly comparable with classic Byzantine Agreement for a number of reasons. First, computing majority is more demanding than agreement on {\em any} input value.
Second, the presence of independent random scheduler may help population protocols in an uneven battle against various adversaries.
Third, one of the emphases of population protocols is on their simplicity and efficiency, in particular, short time and local memory, which in some classic settings are not feasible~to~achieve.
\cbr{As a last remark, we note that the Byzantine node behavior of our model resembles the Byzantine model studied in the cryptography/security literature \cite{Canetti00,Canetti01}.}

%% file: conclusion.tex
\section{Conclusions}
\label{sec:conclusion}

We presented Byzantine-resilient population protocols for the majority 
consensus
problem, 
which can tolerate up to $f = O(n)$ faulty nodes.
The only other known previous work tolerates up to $o(\sqrt n)$ faulty nodes.
Thus, our protocols significantly improve on the number of faulty nodes that can tolerated by 
a $\omega(\sqrt n)$ factor. We introduce a comprehensive range of Byzantine adversary models, which combine variants where the nodes are corrupted statically or dynamically, and the adversary is Full or Weak depending on its capability to observe fully or partially the execution.
For these adversary models we presented lower bounds on the size of the initial differences $d = \Omega(f+1)$.
We also presented three algorithms with different trade-offs between the requirements on $d$,
getting very close to the lower bounds.
The algorithms require $O(\log^3 n)$ states per node, they are fast (polylogarithmically) and correct against the strongest considered Fully Dynamic Byzantine Adversary (unlike the previous algorithm that worked against the Static Byzantine~Adversary).

Compared to the state of the art population protocols that are non-Byzantine-resilient,
our protocols have close asymptotic performance, which is only a poly-log factor away from 
the best known stabilization time, while the number of states is the same ($O(\log^3 n)$).
A question is whether the bounds we presented can be even tighter in the presence of faults. 
Therefore, an interesting direction is to find whether the protocol requirements on $d$ can be tighter,
and if the time and number of states can also asymptotically improve,
to get even closer results to the protocols that do not tolerate faults.
Another interesting open question is whether our techniques,
such as asymmetric/symmetric cancellations, partial/full duplication or combining different protocols by applying random biases, can be used for other related problems
with faulty nodes,
such as multi-valued majority consensus, or even leader election.
Finally, it is interesting to check whether population protocols could improve performance of Byzantine-resilient computation on a {\em multi-hop} communication systems, comparing to the classic message-passing model (c.f.,~\cite{ChlebusKO20}).

\vspace*{1ex}
\noindent
{\bf Acknowledgments:}
This work is supported by grant NSF-CNS-2131538.